  \newcommand*{\CTIV}{Vovk:arXiv0904}

  \newcommand*{\CTXIII}{GTP45}

  \newcommand*{\GTPXXVII}{Vovk:arXiv0905}

\newcommand{\zzrelax}[1]{\relax}

  \documentclass[10pt]{article}
  \usepackage{amsmath,amsthm,amsfonts,amssymb,mathtools,latexsym,graphicx,url,stmaryrd,bm,euscript,fge}
  \newcommand{\Extra}[1]{}

\allowdisplaybreaks[4]

\theoremstyle{plain}
\newtheorem{theorem}{Theorem}

\newtheorem{lemma}[theorem]{Lemma}

\newtheorem*{KMTtheorem}{KMT theorem}

\theoremstyle{definition}
\newtheorem{remark}[theorem]{Remark}

\DeclareMathOperator{\Expect}{\mathbb{E}}
\DeclareMathOperator{\Prob}{\mathbb{P}}

\newcommand{\pretend}[2]{\smash{\mathrlap{#1}}\phantom{#2}}

\DeclareMathOperator{\UEG}{\pretend{\overline{\mathbb{E}}}{\mathbb{E}}^{\rm g}}

\DeclareMathOperator{\UPG}{\pretend{\overline{\mathbb{P}}}{\mathbb{P}}^{\rm g}}

\DeclareMathOperator{\UEM}{\pretend{\overline{\mathbb{E}}}{\mathbb{E}}^{\rm m}}

\DeclareMathOperator{\UPM}{\pretend{\overline{\mathbb{P}}}{\mathbb{P}}^{\rm m}}

\newcommand{\FFF}{\mathcal{F}}
\newcommand{\K}{\EuScript{K}}
\DeclareMathOperator{\ntt}{ntt}

\newcommand{\III}{\boldsymbol{1}}

\newcommand{\given}{\mathbin{|}}
\newcommand{\st}{\mid}

\newcommand{\graph}{\mathrm{graph}}

\def\clap#1{\hbox to 0pt{\hss#1\hss}}

\def\mathrlap{\mathpalette\mathrlapinternal}

\def\mathrlapinternal#1#2{\rlap{$\mathsurround=0pt#1{#2}$}}

  \title{Another example of duality
    between game-theoretic and measure-theoretic probability%
    \thanks{The version of this paper
      at \protect\url{http://probabilityandfinance.com} (Working Paper 46)
      is updated more often.}}
  \author{Vladimir Vovk}

\begin{document}
\maketitle
\begin{abstract}
  This paper makes a small step towards a non-stochastic version
  of superhedging duality relations
  in the case of one traded security with a continuous price path.
  Namely, we prove the coincidence of game-theoretic and measure-theoretic expectation
  for lower semicontinuous positive functionals.
  We consider a new broad definition of game-theoretic probability,
  leaving the older narrower definitions for future work.
\end{abstract}

\section{Introduction}

The words like ``positive'' and ``increasing'' will be understood in the wide sense
(e.g., $a$ is positive if $a\ge0$),
and the qualifier ``strictly'' will indicate the narrow sense
(e.g., $a$ is strictly positive if $a>0$).
The set of all continuous real-valued functions on a topological space $X$ is denoted, as usual, $C(X)$,
and its subset consisting of positive functions is denoted $C^+(X)$.
We abbreviate expressions such as $C([0,T])$ and $C^+([0,T])$, where $T>0$,
$C([0,\infty))$, and $C^+([0,\infty))$
to $C[0,T]$, $C^+[0,T]$, $C[0,\infty)$, and $C^+[0,\infty)$,
respectively,
and let $C_a[0,T]$, $C^+_a[0,T]$, $C_a[0,\infty)$, and $C^+_a[0,\infty)$
stand for the subsets of these sets consisting of the functions $f$ satisfying $f(0)=a$,
for a given constant $a$.

Let $\mathbb{N}:=\{1,2,\ldots\}$ be the set of all strictly positive integers,
and $\mathbb{N}_0:=\{0,1,2,\ldots\}$ be the set of all positive integers.

As usual $a\wedge b$ stands for minimum of $a$ and $b$ and $a\vee b$ for their maximum.
In this paper, the operators $\wedge$ and $\vee$ have higher precedence than the arithmetic operators:
e.g., $a+b\wedge c$ means $a+(b\wedge c)$.
Other conventions of this kind are that:
\begin{itemize}
\item
  Cartesian product $\times$ has higher precedence than union $\cup$;
  so that, e.g., $A\cup\{1\}\times[0,\infty)$
  means $A\cup(\{1\}\times[0,\infty))$;
\item
  implicit multiplication (not using a multiplication sign such as $\times$ or $\cdot$)
  has higher precedence than division;
  so that, e.g., $S/NL$ means $S/(NL)$.
\end{itemize}

In our informal discussions we will use symbols $\approx$ for approximate equality
and $\lesssim$ and $\gtrsim$ for approximate inequalities.

In this paper we consider a finite time interval $[0,T]$ where $T\in(0,\infty)$;
without loss of generality we set $T:=1$.

\section{The main result}

The \emph{sample space} used in this paper, $\Omega:=C^+_1[0,1]$,
is the set of all positive continuous functions $\omega:[0,1]\to[0,\infty)$
such that $\omega(0)=1$.
Intuitively, the functions in $\Omega$ are price paths of a financial security
whose initial price serves as the unit for measuring its later prices.

We equip $\Omega$ with the usual $\sigma$-algebra $\FFF$,
i.e., the smallest $\sigma$-algebra making all functions $\omega\in\Omega\mapsto\omega(t)$,
$t\in[0,1]$, measurable.
A \emph{process} (more fully, an \emph{adapted process}) $\mathfrak{S}$
is a family of extended random variables
$\mathfrak{S}_t:\Omega\to[-\infty,\infty]$, $t\in[0,\infty)$,
such that, for all $\omega,\omega'\in\Omega$ and all $t\in[0,\infty)$,
\[
  \omega|_{[0,t]}
  =
  \omega'|_{[0,t]}
  \Longrightarrow
  \mathfrak{S}_t(\omega)=\mathfrak{S}_t(\omega');
\]
its \emph{sample paths} are the functions $t\in[0,1]\mapsto\mathfrak{S}_t(\omega)$.
A \emph{stopping time} is an extended random variable $\tau:\Omega\to[0,\infty]$ such that,
for all $\omega,\omega'\in\Omega$,
\[
  \omega|_{[0,\tau(\omega)\wedge1]}
  =
  \omega'|_{[0,\tau(\omega)\wedge1]}
  \Longrightarrow
  \tau(\omega)=\tau(\omega'),
\]
where $\omega|_A$ stands for the restriction of $\omega$ to $A\subseteq[0,1]$.
For any stopping time $\tau$,
the $\sigma$-algebra $\FFF_{\tau}$ is defined as the family of all events $E\in\FFF$ such that,
for all $\omega,\omega'\in\Omega$,
\begin{equation}\label{eq:tau}
  \left(
    \omega|_{[0,\tau(\omega)\wedge1]}
    =
    \omega'|_{[0,\tau(\omega)\wedge1]},
    \omega\in E
  \right)
  \Longrightarrow
  \omega'\in E.
\end{equation}
Therefore, a random variable $X$ is $\FFF_{\tau}$-measurable if and only if,
for all $\omega,\omega'\in\Omega$,
\[
  \omega|_{[0,\tau(\omega)\wedge1]}
  =
  \omega'|_{[0,\tau(\omega)\wedge1]}
  \Longrightarrow
  X(\omega)=X(\omega').
\]

\begin{remark}
  Our definitions (convenient for the purposes of this paper)
  are equivalent to the standard ones by Galmarino's test
  (\cite{Dellacherie/Meyer:1978}, IV.100).
\end{remark}

First we define game-theoretic probability and expectation,
partly following Perkowksi and Pr\"omel
\cite{Perkowski/Promel:2013,Perkowski/Promel:2014,Beiglbock/etal:2015}
(this is a ``broad'' definition making our task easier;
the older ``narrow'' definition of \cite{\CTIV}
is much more conservative and might require stronger assumptions
for our main result to hold true;
another broad definition was given in \cite{\CTXIII}).
A \emph{simple trading strategy} $G$
consists of an increasing sequence of stopping times
$\tau_1\le\tau_2\le\cdots$
(we may assume, without loss of generality, $\tau_n\in[0,1]\cup\{\infty\}$
and $\tau_n<\tau_{n+1}$ unless $\tau_n=\infty$)
and, for each $n=1,2,\ldots$, a bounded $\FFF_{\tau_{n}}$-measurable function $h_n$.
It is required that, for each $\omega\in\Omega$,
$\lim_{n\to\infty}\tau_n(\omega)=\infty$.
To such $G$ and an \emph{initial capital} $c\in\mathbb{R}$
corresponds the \emph{simple capital process}
\begin{equation}\label{eq:simple-capital}
  \K^{G,c}_t(\omega)
  :=
  c
  +
  \sum_{n=1}^{\infty}
  h_n(\omega)
  \bigl(
    \omega(\tau_{n+1}(\omega)\wedge t)-\omega(\tau_n(\omega)\wedge t)
  \bigr),
  \quad
  t\in[0,\infty);
\end{equation}
the value $h_n(\omega)$ will be called the \emph{bet}
(or \emph{bet on $\omega$}, or \emph{stake}) at time $\tau_n$,
and $\K^{G,c}_t(\omega)$ will be called the \emph{capital} at time $t$.
For $c\ge0$, let $\mathcal{C}_c$ be the class of positive functionals of the form $\K^{G,c}_1$,
$G$ ranging over simple trading strategies;
intuitively, these are the functionals that can be hedged with initial capital $c$
by a simple strategy that does not risk bankruptcy
(notice that $\forall\omega:\K^{G,c}_1(\omega)\ge0$
implies $\forall\omega\:\forall t:\K^{G,c}_t(\omega)\ge0$).

A class $\mathcal{C}$ of functionals $F:\Omega\to[0,\infty]$ is \emph{$\liminf$-closed}
if $F\in\mathcal{C}$ whenever there is a sequence $F_1,F_2,\ldots$ of functionals in $\mathcal{C}$
such that
\begin{equation}\label{eq:PP}
  \forall\omega\in\Omega:
  F(\omega)
  \le
  \liminf_{n\to\infty}
  F_n(\omega).
\end{equation}
The intuition is that if $F_1,F_2,\ldots$ can be superhedged, so can $F$ in the limit.
It is clear that for each class $\mathcal{C}$ of functionals
there is a smallest $\liminf$-closed class, denoted $\overline{\mathcal{C}}$, containing $\mathcal{C}$.

The \emph{upper game-theoretic expectation} of a functional $F:\Omega\to[0,\infty]$
is defined to be
\begin{equation}\label{eq:UEG}
  \UEG(F)
  :=
  \inf
  \left\{
    c
    \st
    F \in \overline{\mathcal{C}_c}
  \right\}.
\end{equation}
where $\mathcal{C}_c$ is as defined above.
The \emph{upper game-theoretic probability} of $E\subseteq\Omega$
is $\UPG(E):=\UEG(\III_E)$, $\III_E$ being the indicator function of $E$.

The \emph{upper measure-theoretic expectation} of $F$ is defined to be
\begin{equation*}
  \UEM(F)
  :=
  \sup_P
  \int F dP,
\end{equation*}
where $P$ ranges over all \emph{martingale measures},
i.e., probability measures on $\Omega$
under which the process $X_t(\omega):=\omega(t)$ is a martingale\Extra{\ (or local martingale, or supermartingale)},
and $\int$ stands for upper integral.
The \emph{upper measure-theoretic probability} of $E\subseteq\Omega$ is $\UPM(E):=\UEM(\III_E)$.

Now we can state our main result, Theorem~\ref{thm:main},
in which ``lower semicontinuous'' refers to the standard topology on $\Omega$
generated by the usual uniform metric
\begin{equation}\label{eq:metric-Omega}
  \rho_U(\omega,\omega')
  :=
  \sup_{t\in[0,1]}
  \lvert \omega(t)-\omega'(t) \rvert.
\end{equation}

\begin{theorem}\label{thm:main}
  For any lower semicontinuous functional $F:\Omega\to[0,\infty]$,
  \begin{equation*}
    \UEG(F)
    =
    \UEM(F)
  \end{equation*}
  (the inequality $\ge$ holding for all $F:\Omega\to[0,\infty]$).
\end{theorem}

An earlier result of the same kind
is the discrete-time Theorem~1 of \cite{\GTPXXVII}.

\section{Proof of Theorem~2}

In this section we prove the coincidence of $\UEG$ and $\UEM$
on ``simple'' (lower semicontinuous in this version of the paper) positive functionals.
We prove the inequality $\ge$ in Subsection~\ref{subsec:ge}
and the inequality $\le$ in Subsection~\ref{subsec:le}.
Notice that we can ignore $\omega\in\Omega$ such that $0=\omega(t)<\omega(s)$
for some $0\le t<s$.

On a few occasions we will use the following simple lemma.
\begin{lemma}\label{lem:outer}
  The functions $\UEG$ and $\UEM$ are $\sigma$-subadditive:
  for any sequence of positive functionals $F_1,F_2,\ldots$ (taking values in $[0,\infty]$),
  \begin{align}
    \UEG
    \left(
      \sum_{n=1}^{\infty} F_n
    \right)
    &\le
    \sum_{n=1}^{\infty}
    \UEG(F_n),\label{eq:subadditivity-1}\\
    \UEM
    \left(
      \sum_{n=1}^{\infty} F_n
    \right)
    &\le
    \sum_{n=1}^{\infty}
    \UEM(F_n).\label{eq:subadditivity-2}
  \end{align}
  (And therefore, the set functions $\UPG$ and $\UPM$ are outer measures.)
\end{lemma}
\begin{proof}
  We can deduce \eqref{eq:subadditivity-2} from the $\sigma$-subadditivity of $F\mapsto\int F dP$:
  indeed, for each $\epsilon>0$,
  \begin{align*}
    \UEM
    \left(
      \sum_{n=1}^{\infty} F_n
    \right)
    &=
    \sup_P
    \int
      \sum_{n=1}^{\infty} F_n
    dP
    \le
    \int
      \sum_{n=1}^{\infty} F_n
    dP_0
    +
    \epsilon
    \le
    \sum_{n=1}^{\infty}
    \int
      F_n
    dP_0
    +
    \epsilon\\
    &\le
    \sum_{n=1}^{\infty}
    \sup_P
    \int
      F_n
    dP
    +
    \epsilon
    =
    \sum_{n=1}^{\infty}
    \UEM(F_n)
    +
    \epsilon,
  \end{align*}
  where $P_0$ is a martingale measure.

  As for \eqref{eq:subadditivity-1},
  we start from a new definition of $\overline{\mathcal{C}_c}$.
  Define $\mathcal{C}_c^{\alpha}$ by transfinite induction over the countable ordinals $\alpha$
  (see, e.g., \cite{Dellacherie/Meyer:1978}, 0.8) as follows:
  \begin{itemize}
  \item
    $\mathcal{C}_c^0:=\mathcal{C}_c$;
  \item
    for $\alpha>0$,
    $F\in\mathcal{C}_c^{\alpha}$ if and only if there exists a sequence $F_1,F_2,\ldots$
    of functionals in $\mathcal{C}_c^{<\alpha}:=\cup_{\beta<\alpha}\mathcal{C}_c^{\beta}$
    such that \eqref{eq:PP} holds.
  \end{itemize}
  It is easy to check that $\overline{\mathcal{C}_c}$ is the union of the nested family $\mathcal{C}_c^{\alpha}$
  over all countable ordinals $\alpha$.

  First we prove finite subadditivity
  (\eqref{eq:subadditivity-1} with $\infty$ replaced by a natural number),
  which will immediately follow from
  \[
    \left(
      F_i\in\overline{\mathcal{C}_{c_i}},\;
      i=1,\ldots,n
    \right)
    \Longrightarrow
    \left(
      \sum_{i=1}^n F_i \in \overline{\mathcal{C}_{\sum_{i=1}^n c_i}}
    \right).
  \]
  It suffices to prove, for each countable ordinal $\alpha$,
  \begin{equation}\label{eq:alpha}
    \left(
      F_i\in\mathcal{C}^{\alpha}_{c_i},\;
      i=1,\ldots,n
    \right)
    \Longrightarrow
    \left(
      \sum_{i=1}^n F_i \in \mathcal{C}_{\sum_{i=1}^n c_i}^{\alpha}
    \right)
  \end{equation}
  (this is the implication that we will actually need below).
  This is true for $\alpha=0$ (by the definition of a simple trading strategy),
  so we fix a countable ordinal $\alpha>0$
  and assume that the statement holds for all ordinals below $\alpha$.
  Let us also assume the antecedent of \eqref{eq:alpha}.
  For each $i\in\{1,\ldots,n\}$ let $F_i^j\in\mathcal{C}_c^{<\alpha}$, $j=1,2,\ldots$,
  be a sequence such that
  \begin{equation*}
    \forall\omega\in\Omega:
    F_i(\omega)
    \le
    \liminf_{j\to\infty}
    F_i^j(\omega).
  \end{equation*}
  For each $j$, the inductive assumption gives
  \[
    \sum_{i=1}^n F_i^j
    \in
    \mathcal{C}^{<\alpha}_{\sum_{i=1}^n c_i}
  \]
  (since there are finitely many $i$, there is $\beta=\beta_j<\alpha$
  such that $F_i^j\in\mathcal{C}^{\beta}_{c_i}$ for all $i\in\{1,\ldots,n\}$).
  By the definition of $\mathcal{C}^{\alpha}$,
  \[
    \liminf_{j\to\infty}
    \sum_{i=1}^n F_i^j
    \in
    \mathcal{C}^{\alpha}_{\sum_{i=1}^n c_i},
  \]
  which implies, by the Fatou lemma,
  \[
    \sum_{i=1}^n
    \liminf_{j\to\infty}
    F_i^j
    \in
    \mathcal{C}^{\alpha}_{\sum_{i=1}^n c_i},
  \]
  which in turn implies
  \[
    \sum_{i=1}^n F_i
    \in
    \mathcal{C}^{\alpha}_{\sum_{i=1}^n c_i}.
  \]

  The countable subadditivity \eqref{eq:subadditivity-1}
  now follows immediately from Lemma~\ref{lem:Fatou} below:
  \begin{multline*}
    \UEG
    \left(
      \sum_{n=1}^{\infty} F_n
    \right)
    =
    \UEG
    \left(
      \liminf_{N\to\infty}
      \sum_{n=1}^{N} F_n
    \right)
    \le
    \liminf_{N\to\infty}
    \UEG
    \left(
      \sum_{n=1}^{N} F_n
    \right)\\
    \le
    \liminf_{N\to\infty}
    \sum_{n=1}^{N}
    \UEG(F_n)
    =
    \sum_{n=1}^{\infty}
    \UEG(F_n).
    \qquad\qed
  \end{multline*}
  \renewcommand{\qedsymbol}{}
\end{proof}

\begin{remark}
  The original ``broad'' definition of game-theoretic probability and expectation
  in \cite{Perkowski/Promel:2013} is given by \eqref{eq:UEG} with $\mathcal{C}_c^1$
  in place of $\overline{\mathcal{C}_c}$.
\end{remark}

The following lemma (already used in the proof of Lemma~\ref{lem:outer} above)
is the analogue of the Fatou lemma for the broad definition of game-theoretic probability.
\begin{lemma}\label{lem:Fatou}
  For any sequence of positive functionals $F_1,F_2,\ldots$,
  \begin{equation}\label{eq:Fatou}
    \UEG
    \left(
      \liminf_{n\to\infty} F_n
    \right)
    \le
    \liminf_{n\to\infty}
    \UEG(F_n).
  \end{equation}
\end{lemma}
\begin{proof}
  Let $c$ be the right-hand side of \eqref{eq:Fatou} and $\epsilon>0$.
  There is a strictly increasing sequence $n_1<n_2<\cdots$ such that $\UEG(F_{n_i})<c+\epsilon$ for all $i$.
  Since $F_{n_i}\in\overline{\mathcal{C}_{c+\epsilon}}$ for all $i$,
  we have $\liminf_{i\to\infty}F_{n_i}\in\overline{\mathcal{C}_{c+\epsilon}}$,
  which implies $\liminf_{n\to\infty}F_{n}\in\overline{\mathcal{C}_{c+\epsilon}}$,
  which in turn implies $\UEG(\liminf_{n\to\infty}F_{n})\le c+\epsilon$.
  Since $\epsilon$ can be made arbitrarily small, this completes the proof.
\end{proof}

\subsection{Inequality $\ge$}\label{subsec:ge}

The goal of this subsection is to prove
\begin{equation}\label{eq:m-le-g}
  \UEM(F)
  \le
  \UEG(F)
\end{equation}
for all functionals $F:\Omega\to[0,\infty]$
(we will not need the assumptions that $F$ is bounded or measurable).

First we will prove
\begin{equation}\label{eq:goal}
  \Expect_P
  \left(
    \K^{G,c}_1
    -
    c
  \right)
  \le
  0
\end{equation}
for all martingale measures $P$,
where $G$ is a simple trading strategy
whose stopping times and bets will be denoted $\tau_1,\tau_2,\ldots$
and $h_1,h_2,\ldots$, respectively,
and $c$ is an initial capital.
Fix such a $P$.
By the Fatou lemma
(applied to the partial sums in \eqref{eq:simple-capital}),
it suffices to prove \eqref{eq:goal}
assuming that the sequence of stopping time is finite:
$\tau_n=\infty$ for all $n>N$ for a given $N\in\mathbb{N}$
(which in turn implies that the bets $h_n$ are bounded in absolute value by a given constant).

For each $k=1,2,\ldots$, set $\tau^k_n:=2^{-k}\lceil 2^k\tau_n\rceil$
and let $\mathfrak{S}^k$ be the simple capital process
corresponding to initial capital $\mathfrak{S}_0^k=c$,
stopping times $\tau^k_n$, and bets $h_n$
(remember that our definition of a simple trading strategy allows $\tau_n=\tau_{n+1}$).
It is easy to check that, for all $k$ and $n=0,\ldots,2^k-1$,
\begin{equation}\label{eq:easy}
  \Expect_P
  \left(
    \mathfrak{S}^k_{(n+1)2^{-k}}
    -
    \mathfrak{S}^k_{n2^{-k}}
  \right)
  =
  0;
\end{equation}
indeed, the difference
$
  \mathfrak{S}^k_{(n+1)2^{-k}}
  -
  \mathfrak{S}^k_{n2^{-k}}
$
is the product of the bounded $\FFF_{n2^{-k}}$-measurable function
\[
  h
  :=
  \sum_{i=1}^N
  h_i
  \III_{\{\tau^k_i=n2^{-k},\tau^k_{i+1}>n2^{-k}\}}
\]
and the martingale difference
$
  \omega((n+1)2^{-k})
  -
  \omega(n2^{-k})
$,
and so
\begin{align*}
  &\Expect_P
  \left(
    \mathfrak{S}^k_{(n+1)2^{-k}}
    -
    \mathfrak{S}^k_{n2^{-k}}
  \right)\\
  &=
  \Expect_{P(d\omega)}
  \left(
    \Expect_{P(d\omega)}
    \left(
      h(\omega)
      \left(
        \omega((n+1)2^{-k})
        -
        \omega(n2^{-k})
      \right)
      \st
      \FFF_{n2^{-k}}
    \right)
  \right)\\
  &=
  \Expect_{P(d\omega)}
  \left(
    h(\omega)
    \Expect_{P(d\omega)}
    \left(
      \left(
        \omega((n+1)2^{-k})
        -
        \omega(n2^{-k})
      \right)
      \st
      \FFF_{n2^{-k}}
    \right)
  \right)
  =
  0.
\end{align*}
Summing \eqref{eq:easy} over $n$ (of which there are finitely many),
\[
  \Expect_P
  \left(
    \mathfrak{S}^k_1
    -
    c
  \right)
  =
  0,
\]
which in turn implies, by the Fatou lemma, \eqref{eq:goal}.

We will complete the proof of \eqref{eq:m-le-g} by transfinite induction,
as in Lemma~\ref{lem:outer}.
Rewrite \eqref{eq:m-le-g} as $\UEM(F)\le c$ for all $F\in\overline{\mathcal{C}_c}$.
Fix $c$ and $F\in\overline{\mathcal{C}_c}$.
In the previous paragraph we checked that $\UEM(F)\le c$ if $F\in\mathcal{C}_c^0$.
Therefore, it remains to prove, for a given countable ordinal $\alpha>0$,
that $\UEM(F)\le c$
assuming that $F\in\mathcal{C}^{\alpha}_c$
and that $\UEM(G)\le c$ for all $G\in\mathcal{C}_c^{<\alpha}$.
Let $F_n\in\mathcal{C}_c^{<\alpha}$, $n=1,2,\ldots$,
be a sequence of functionals such that $F\le\liminf_n F_n$.
Suppose $\UEM(F)>c$ and find a martingale measure $P$ such that $c<\int F dP$.
We get a contradiction by the Fatou lemma and the inductive assumption:
\[
  c
  <
  \int F dP
  \le
  \int
  \liminf_{n\to\infty}
  F_n
  dP
  \le
  \liminf_{n\to\infty}
  \int F_n dP
  \le
  \liminf_{n\to\infty}
  c
  =
  c.
\]

\subsection{Inequality $\le$}\label{subsec:le}

In this section we will prove that
\begin{equation}\label{eq:g-le-m}
  \UEG(F)
  \le
  \UEM(F).
\end{equation}
Since $\UEG(F)$ is defined as an infimum and $\UEM(F)$ as a supremum,
it suffices to construct a martingale measure $P$ and a superhedging capital process
for a given lower semicontinuous positive functional $F$
such that $\int F dP$ is close to (or greater than) the initial capital of the process.

\subsubsection{Reductions I}

The goal of this section is to show that, without loss of generality,
we can assume that the functional $F$ is bounded and lower semicontinuous in a stronger sense.

For a general lower semicontinuous $F:\Omega\to[0,\infty]$ and $n\in\mathbb{N}$,
set $F_n(\omega):=F(\omega)\wedge n$.
Assuming $\UEG(F_n)\le\UEM(F_n)$ for all $n$, let us prove $\UEG(F)\le\UEM(F)$.
Set $c:=\UEM(F)+\epsilon$ for a small $\epsilon>0$.
Since
\[
  \UEG(F_n)\le\UEM(F_n)\le\UEM(F)<c,
\]
we have $F_n\in\overline{\mathcal{C}_c}$.
Since $\overline{\mathcal{C}_c}$ is $\liminf$-closed,
we have
\[
  F
  =
  \liminf_{n\to\infty}
  F_n
  \in
  \overline{\mathcal{C}_c}
\]
and, therefore, $\UEG(F)\le c$.
Since $\epsilon$ can be arbitrarily small,
this completes the proof of $\UEG(F)\le\UEM(F)$.
Therefore, we can, and will, assume that $F$ is bounded above.

In the rest of this paper, instead of the uniform metric \eqref{eq:metric-Omega}
we will consider the Hausdorff metric
\begin{multline}\label{eq:Hausdorff}
  \rho_H(\omega,\omega')
  :=
  H(\bar\omega,\bar{\omega'})
  :=
  \adjustlimits
  \sup_{(t,x)\in\bar\omega}
  \inf_{(t',x')\in\bar{\omega'}}
  \left\|
    \left(
      t-t',x-x'
    \right)
  \right\|\\
  \vee
  \adjustlimits
  \sup_{(t',x')\in\bar{\omega'}}
  \inf_{(t,x)\in\bar\omega}
  \left\|
    \left(
      t-t',x-x'
    \right)
  \right\|,
\end{multline}
where $\left\|\cdot\right\|=\left\|\cdot\right\|_{\infty}$
stands for the $\ell_{\infty}$ norm
$\left\|(a,b)\right\|:=\left|a\right|\vee\left|b\right|$
in $\mathbb{R}^2$
and each element $\omega$ of $\Omega$ is mapped to the set
$\bar\omega\subseteq[0,1]\times[0,\infty)$
defined to be the union $\graph(\omega)\cup\{1\}\times[0,\infty)$
of the graph of $\omega$ and the ray $\{1\}\times[0,\infty)$.

\begin{remark}\label{rem:Hausdorff}
  Notice that the metrics $\rho_U$ and $\rho_H$ lead to different topologies:
  e.g., there is an unbounded sequence $\omega_n$ of elements of $\Omega$
  such that $\omega_n\to0$ in $\rho_H$.
  The $\ell_{\infty}$ norm (used in our definition of $\rho_H$) is, of course,
  equivalent to the Euclidean norm $\ell_2$,
  but sometimes it leads to slightly simpler formulas.
  An example of a functional $F:\Omega\to\mathbb{R}$ continuous in $\rho_H$
  is $F(\omega):=F'(\omega|_{[0,1-\epsilon]})$,
  where $F'$ is a functional on $C^+_1[0,1-\epsilon]$ continuous in the uniform metric
  and $\epsilon\in(0,1)$ is a strictly positive constant.
\end{remark}

\begin{remark}
  On the other hand,
  the topologies generated by the metrics $\rho_U$ and $\rho_H$
  lead to the same Borel $\sigma$-algebra.
  Since the topology generated by $\rho_U$ is finer than the one generated by $\rho_H$,
  it suffices to check that every $\rho_U$-Borel set is a $\rho_H$-Borel set.
  Since the $\rho_U$-topology is separable,
  it suffices to check that every open ball in $\rho_U$ is a $\rho_H$-Borel set.
  This is easy; moreover,
  every open ball in $\rho_U$ is the intersection of a sequence of $\rho_H$-open sets.
\end{remark}

Let us check that in Theorem~\ref{thm:main} we can further assume that
$F$ is lower semicontinuous in the Hausdorff topology on $\Omega$
(this observation develops the end of Remark~\ref{rem:Hausdorff}).
Suppose that Theorem~\ref{thm:main} holds for all (bounded) positive functionals
that are lower semicontinuous in the Hausdorff topology.
It is clear that we can replace the sample space $\Omega$
by the sample space $\Omega^*:=C^+_1[0,2]$;
let us do so.
Now let $F$ be a lower semicontinuous (in the usual uniform topology) positive functional on $\Omega$.
Define
\[
  F^*(\omega)
  :=
  F(\omega|_{[0,1]}),
  \quad
  \omega\in\Omega^*.
\]
Then $F^*$ is lower semicontinuous in the Hausdorff metric on $\Omega^*$
(defined by \eqref{eq:Hausdorff} where the ray $\{1\}\times[0,\infty)$ in the definition of $\bar\omega$
is replaced by $\{2\}\times[0,\infty)$).
Indeed, for any constant $c$, the set $\{F^*>c\}$ is open:
if $\omega_n\to\omega$ in the Hausdorff metric on $C^+_1[0,2]$,
then $\omega_n|_{[0,1]}\to\omega|_{[0,1]}$ in the usual topology
(had $\omega_n|_{[0,1]}$ not converged to $\omega|_{[0,1]}$ in the usual topology,
we could have found $\epsilon>0$ and $t_n\in[0,1]$
such that $\left|\omega_n(t_n)-\omega(t_n)\right|>\epsilon$ for infinitely many $n$
and arrived at a contradiction by considering a limit point of those $t_n$),
and so
\[
  \forall n:F^*(\omega_n)\le c
  \Longleftrightarrow
  \forall n:F(\omega_n|_{[0,1]})\le c
  \Longrightarrow
  F(\omega|_{[0,1]})\le c
  \Longleftrightarrow
  F^*(\omega)\le c.
\]
Therefore, our assumption
(the non-trivial part of Theorem~\ref{thm:main} for the Hausdorff metric) gives
\begin{equation*}
  \UEG(F^*)
  \le
  \UEM(F^*),
\end{equation*}
and it suffices to prove $\UEG(F)\le\UEG(F^*)$ and $\UEM(F^*)\le\UEM(F)$.

First let us check that $\UEG(F)\le\UEG(F^*)$.
This follows from the class $\overline{\mathcal{C}_c}$
dominating the class $\overline{\mathcal{C}^*_c}$
for all $c>0$,
where the class $\mathcal{C}_c$ is as defined above,
the class $\mathcal{C}^*_c$ is the analogue of this class
for the time interval $[0,2]$ rather than $[0,1]$,
and a class $\mathcal{A}$ of functionals on $\Omega$ is said to \emph{dominate}
a class $\mathcal{B}$ of functionals on $\Omega^*$
if for any $G\in\mathcal{B}$ there exists $G'\in\mathcal{A}$
that \emph{dominates} $G$ in the sense that,
for any $\omega\in\Omega$,
\[
  G'(\omega)
  \ge
  \Expect_{W(d\xi)}(G(\omega\xi))
\]
where $W$ is the Wiener measure on $C_0[0,1]$
and $\omega\xi:[0,2]\to[0,\infty)$ is the continuous combination of $\omega$ and $\xi$
defined as follows:
\[
  (\omega\xi)(t)
  :=
  \begin{cases}
    \omega(t) & \text{if $t\in[0,1]$}\\
    \omega(1)+\xi(t-1) & \text{if $t\in[1,\tau]$}\\
    0 & \text{if $t\in(\tau,1]$}
  \end{cases}
\]
where
\[
  \tau
  :=
  \inf\{t\in[1,2]\st\omega(1)+\xi(t-1)=0\}
\]
with $\inf\emptyset:=2$.
Indeed, assuming that $\overline{\mathcal{C}_c}$ dominates $\overline{\mathcal{C}^*_c}$
for all $c>0$,
we obtain
\begin{equation*}
  \UEG(F)
  =
  \inf
  \left\{
    c
    \st
    F \in \overline{\mathcal{C}_c}
  \right\}
  \le
  \inf
  \left\{
    c
    \st
    F^* \in \overline{\mathcal{C}^*_c}
  \right\}
  =
  \UEG(F^*),
\end{equation*}
where the inequality follows from the fact that,
whenever $F^* \in \overline{\mathcal{C}^*_c}$,
$F^*$ is dominated by some $G\in\overline{\mathcal{C}_c}$,
which implies
\[
  \forall\omega\in\Omega:
  F(\omega)
  =
  \Expect_{W(d\xi)}(F^*(\omega\xi))
  \le
  G(\omega),
\]
which in turn implies $F\in\overline{\mathcal{C}_c}$.
Therefore, it remains to prove that $\overline{\mathcal{C}_c}\sqsupseteq\overline{\mathcal{C}^*_c}$,
where $\mathcal{A}\sqsupseteq\mathcal{B}$ stands for ``$\mathcal{A}$ dominates $\mathcal{B}$''.
Let us fix $c>0$.
Our proof is by transfinite induction.
The basis of induction $\mathcal{C}_c^0\sqsupseteq\mathcal{C}_c^{*0}$
follows from the fact that $\K^{G,c}_2$ is always dominated by $\K^{G,c}_1$:
indeed, for a fixed $\omega\in\Omega$,
any simple capital process $\K^{G,c}_t(\omega\xi)$ over $\Omega^*$
is a supermartingale over $t\in[1,2]$
(see the beginning of the proof of Lemma~6.4 in \cite{\CTIV}),
where $\xi\sim W$, as above.
It remains to prove that $\mathcal{C}_c^{\alpha}\sqsupseteq\mathcal{C}_c^{*\alpha}$
for each countable ordinal $\alpha>0$
assuming $\mathcal{C}_c^{\beta}\sqsupseteq\mathcal{C}_c^{*\beta}$
for each $\beta<\alpha$.
Let us make this assumption and let $G\in\mathcal{C}_c^{*\alpha}$.
Find a sequence of functionals $G_n\in\mathcal{C}_c^{*<\alpha}$, $n=1,2,\ldots$,
such that $G\le\liminf_{n\to\infty}G_n$.
By the inductive assumption,
for each $n$ there is $G'_n\in\mathcal{C}_c^{<\alpha}$ that dominates $G$.
By the Fatou lemma we now have, for each $\omega\in\Omega$,
\begin{multline*}
  \Expect_{W(d\xi)}(G(\omega\xi))
  \le
  \Expect_{W(d\xi)}(\liminf_{n\to\infty}G_n(\omega\xi))\\
  \le
  \liminf_{n\to\infty}\Expect_{W(d\xi)}(G_n(\omega\xi))
  \le
  \liminf_{n\to\infty}G'_n(\omega).
\end{multline*}
In other words, $G':=\liminf_{n\to\infty}G'_n\in\mathcal{C}_c^{\alpha}$ dominates $G$.
This completes the proof of $\UEG(F)\le\UEG(F^*)$.

To check that $\UEM(F)\ge\UEM(F^*)$, i.e.,
\[
  \sup_P
  \int F dP
  \ge
  \sup_{P^*}
  \int F^* dP^*,
\]
where $P$ ranges over the martingale measures on $\Omega$
and $P^*$ over the martingale measures on $\Omega^*$,
it suffices to notice that for any $P^*$ we can take as $P$ the martingale measure defined by
\[
  P(E)
  :=
  P^*
  \left(
    \bigl\{
      \omega\in\Omega^*
      \st
      \omega_{[0,1]}\in E
    \bigr\}
  \right)
\]
for all measurable $E\subseteq\Omega$
(essentially, the restriction of $P^*$ to cylinder sets in $\Omega^*$).

From now on $F$ is assumed bounded and lower semicontinuous in the Hausdorff metric.

\subsubsection{Reductions II}

We further simplify the functional $F$
analogously to the series of reductions in \cite{\CTIV}, Section~10.
We will modify the notation of \cite{\CTIV}
and write $\tilde\omega$ for $\ntt(\omega)$
(as defined in Section~5 of \cite{\CTIV})
and $\phi_s$ for $\tau_s$ (also defined in Section~5 of \cite{\CTIV}).
Let the domain of $\tilde\omega$ be $[0,D(\omega)]$ or $[0,D(\omega))$\label{p:D}
(it has this form for typical $\omega\in\Omega$).

Let $\Omega''$ be the family of all sets of the form $A\cup\{1\}\times[0,\infty)$
where $A\subseteq[0,1]\times[0,\infty)$ is a bounded closed set
and $\Omega'\subseteq\Omega''$ be the set of all $A\in\Omega''$ satisfying
\begin{itemize}
\item
  each vertical cut $A^t:=\{a\st(t,a)\in A\}\subseteq[0,\infty)$ of $A$, where $t\in[0,1)$,
  is non-empty and connected (i.e., is a closed interval);
\item
  $A^0\ni 1$ (and, automatically, $A^1=[0,\infty)$).
\end{itemize}

\begin{lemma}
  The set $\Omega'$ is closed in $\Omega''$
  (equipped with the Hausdorff metric).
\end{lemma}

\begin{proof}
  Let $A_n\to A$ for some $A_n\in\Omega'$, $n=1,2,\ldots$, and $A\in\Omega''$;
  our goal is to prove $A\in\Omega'$.
  Let $B>0$ be such that $A\subseteq[0,1)\times[0,B]\cup\{1\}\times[0,\infty)$.
  First we check that each cut of $A$ is non-empty:
  indeed, suppose $A^t=\emptyset$ for $t\in(0,1)$ (the case $t\in\{0,1\}$ is trivial);
  since $[0,B]$ is compact, this implies $A^{t'}=\emptyset$ for all $t'$ in a neighbourhood of $t$;
  therefore, $(A')^t=\emptyset$ for all $A'$ in a Hausdorff neighbourhood of $A$.
  Now suppose there is $t\in[0,1)$
  (the case $t=0$ will be also covered by our argument)
  such that $A^t$ is not connected, say $A^t$ contains points both above and below $b\in(0,B)\setminus A^t$.
  Let $O$ be a connected open neighbourhood of $t$ and $\delta$ be a strictly positive constant
  such that, for all $s\in O$,
  $A^s$ contains points below $b-\delta$ or points above $b+\delta$
  but does not contain points in $(b-\delta,b+\delta)$.
  Choose another connected open neighbourhood $O'$ of $t$ such that $\overline{O'}\subseteq O$.
  Let $O_n^+$ be the set of $s\in O'$ such that $A_n^s$ contains points above $b$
  and $O_n^-$ be the set of $s\in O'$ such that $A_n^s$ contains points below $b$.
  Since, for sufficiently large $n$,
  $O_n^+$ and $O_n^-$ are disjoint sets that are closed in $O'$
  (closed in $O'$ by the compactness of $[0,b]$ and $[b,B]$)
  and $O'$ is connected,
  either $O'=O_n^+$ or $O'=O_n^-$.
  This makes $A_n\to A$ impossible.
  The remaining condition, $A^0\ni 1$,
  is obvious.
\end{proof}

Now it is easy to see that $\Omega'$ is the closure of $\bar\Omega:=\{\bar\omega\st\omega\in\Omega\}$
in $\Omega''$.

We extend the functional $F$ to the set $\Omega'$ by setting
\begin{equation*}
  F'(A)
  :=
  \liminf_{\omega\fgerightarrow A}
  F(\omega),
\end{equation*}
where $\omega$ ranges over $\Omega$
and $\omega\fgerightarrow A$ is the convergence in the sense of the ``one-sided Hausdorff metric''
(defined in terms of $\ell_{\infty}$, as always in this paper):
namely, the $\epsilon$-neighbourhood of $A$ is the set of $\omega\in\Omega$ such that
\[
  \adjustlimits
  \sup_{(t,a)\in\graph(\omega)}
  \inf_{(t',a')\in A}
  \left|t-t'\right|
  \vee
  \left|a-a'\right|
   <
   \epsilon,
\]
and
$
  \liminf_{\omega\fgerightarrow A}
  F(\omega)
$
is the limit of the infimum of $F$ over the $\epsilon$-neighbourhood of $A$ as $\epsilon\to0$.
Since no $\epsilon$-neighbourhood of $A\in\Omega'$ is empty for $\epsilon>0$
(see Lemma~\ref{lem:non-empty} below),
$F'$ takes values in $[0,\sup F]$.
Notice that $F'$ is monotonic: $F'(A)\ge F'(B)$ when $A\subseteq B$.

\begin{lemma}\label{lem:non-empty}
  Let $A\in\Omega'$ and $\epsilon>0$.
  The $\epsilon$-neighbourhood of $A$ is not empty.
\end{lemma}

\begin{proof}
  Draw parallel vertical lines $t=i/n$, $i=0,\ldots,n$,
  at regular intervals in the semi-infinite region $[0,1]\times[0,\infty)$
  of the $(t,a)$-plane starting from $t=0$ and ending at $t=1$;
  the interval $1/n$ between the lines should be at most $\epsilon$:
  $1/n\le\epsilon$.
  Similarly, draw parallel horizontal lines $a=i/n$, $i=0,1,\ldots$,
  at regular intervals in the same semi-infinite region $[0,1]\times[0,\infty)$
  starting from $a=0$.
  The region $[0,1]\times[0,\infty)$ will be split into squares of size at most $\epsilon\times\epsilon$;
  these squares can be partitioned into columns
  (each column consisting of squares with equal $t$-coordinates).
  Let us mark the squares whose intersection with $A$ is non-empty.
  It suffices to prove that in each column the marked squares
  form a contiguous array
  and that these arrays overlap for each pair of adjacent columns:
  indeed, in this case we will be able to travel in a continuous manner
  from the point $(0,1)$   to the line $t=1$ via marked squares.

  Suppose there is an unmarked square
  such that there is a point $(t',a')\in A$ in a square below it (in the same column)
  and there is a point $(t'',a'')\in A$ in a square above it (in the same column).
  (Notice that this unmarked square cannot be in the right-most column,
  and so the column containing the unmarked square can be regarded as bounded
  since $A$ is bounded, apart from the line $t=1$.)
  Suppose, for concreteness, $t'<t''$.
  All $t\in[t',t'']$ are now split into two disjoint closed sets:
  those for which there are $(t,a)\in A$ for $a$ above the unmarked square
  and those for which there are $(t,a)\in A$ for $a$ below the unmarked square.
  Since $[t',t'']$ is connected,
  one of those disjoint closed sets is empty,
  and we have arrived at a contradiction.

  Now it is obvious that the arrays of marked squares overlap for each pair of adjacent columns:
  remember that the intersection of $A$ with the vertical line between the two columns
  is non-empty and connected.
\end{proof}

Let us check that $F':\Omega'\to[0,\infty)$ is lower semicontinuous
and that $F'(\bar\omega)=F(\omega)$ for all $\omega\in\Omega$;
the latter property can be written as $F'|_{\Omega}=F$,
where $F'|_{\Omega}:\Omega\to[0,\infty)$
is defined by $F'|_{\Omega}(\omega):=F'(\bar\omega)$.
Indeed:
\begin{itemize}
\item
  Let $c:=F'(A)$ and $\epsilon>0$;
  we are required to prove that $F'(B)\ge c-\epsilon$ for all $B$ in an open Hausdorff ball around $A$.
  Let $\delta>0$ be so small that $F(\omega)>c-\epsilon$ for all $\omega\in\Omega$
  in the $\delta$-neighbourhood of $A$.
  Let $B$ be in the open $\delta/2$-ball around $A$ (in the sense of the Hausdorff metric).
  If $\omega$ is in the $\delta/2$-neighbourhood of $B$,
  then $\omega$ will be in the $\delta$-neighbourhood of $A$,
  and so $F(\omega)>c-\epsilon$.
  Therefore, for such $B$ we have $F'(B)\ge c-\epsilon$.
\item
  Let $\omega\in\Omega$.
  We have $F'(\bar\omega)\le F(\omega)$ since $\omega$ is in the $\epsilon$-neighbourhood of $\bar\omega$
  for any $\epsilon>0$.
  And the inequality $F'(\bar\omega)\ge F(\omega)$ follows from the lower semicontinuity of $F$ on $\Omega$
  (in the metric $\rho_H$)
  and the fact that $\omega_n\fgerightarrow\bar\omega$ implies $\rho_H(\omega_n,\omega)\to0$.
  To check the last statement, suppose that there is a subsequence of $\omega_n$
  such that $\rho_H(\omega_{n},\omega)\ge\epsilon$ for the subsequence,
  where $\epsilon>0$;
  without loss of generality we can assume that for each element of the subsequence
  there is a point $(t_n,a_n)\in\graph(\omega)$ such that $t_n\le1-\epsilon$
  and there are no points of $\graph(\omega_n)$
  in the square $[t_n-\epsilon,t_n+\epsilon]\times[a_n-\epsilon,a_n+\epsilon]$.
  Let $(t,a)$ be a limit point of $(t_n,a_n)$, which obviously exists and belongs to $\graph(\omega)$.
  There is another subsequence of $\omega_n$ for which there are no points of $\graph(\omega_n)$
  in the square $[t-\epsilon/2,t+\epsilon/2]\times[a-\epsilon/2,a+\epsilon/2]$.
  This contradicts $\omega_n\fgerightarrow\bar\omega$:
  the distance from $(t,\omega_n(t))$ to any point of $\graph(\omega)$
  stays above a strictly positive constant as $n\to\infty$.
\end{itemize}

Let us now check that we can assume $F=F'|_{\Omega}$ where $F':\Omega'\to[0,\infty)$
is continuous (in the Hausdorff metric).
First suppose \eqref{eq:g-le-m} holds for the restrictions to $\Omega$
of all continuous functions of the type $\Omega'\to[0,\infty)$,
but we are given $F=F'|_{\Omega}$ for $F'$ that is only lower semicontinuous.
Each lower semicontinuous function on a metric space
(such as $\Omega'$ with the Hausdorff metric)
is the limit of an increasing sequence of continuous functions (see, e.g., \cite{Engelking:1989},
1.7.15(c)),
so we can find an increasing sequence of continuous functionals $F_n\nearrow F'$ on $\Omega'$.
Let $\epsilon>0$.
For each $n$, by assumption we have $F_n|_{\Omega}\in\overline{\mathcal{C}_c}$
where $c:=\UEM(F)+\epsilon>\UEM(F_n|_{\Omega})$.
Since $\overline{\mathcal{C}_c}$ is $\liminf$-closed,
\[
  F
  =
  \liminf_{n\to\infty}
  F_n|_{\Omega}
  \in
  \overline{\mathcal{C}_c}.
\]
Therefore, $\UEG(F)\le c=\UEM(F)+\epsilon$ and so, since $\epsilon$ can be arbitrarily small,
$\UEG(F)\le\UEM(F)$.

Let us check that we can replace our new assumption of continuity by the assumption that
$F$ depends on $\omega\in\Omega$
only via the values $\tilde\omega(iS/N)$ and $\phi_{iS/N}(\omega)$,
$i=1,\ldots,N$ (remember that we are interested in the case $\tilde\omega(0)=\omega(0)=1$),
for some $S>0$ and some $N\in\mathbb{N}$
(in particular, only via $\tilde\omega|_{[0,S]}$ and $\phi(\omega)|_{[0,S]}$).
We ignore events of zero upper game-theoretic probability
(such as the event that $\tilde\omega$ does not exist).
Let $\epsilon>0$ and let $S$ and $N$ be sufficiently large
(we will explain later how large $S$ and $N$ should be for a given $\epsilon$).
Let $A_1\subseteq\Omega$
consist of all $\omega\in\Omega$ such that $D(\omega)>S$
($D(\omega)$ is defined at the beginning of this subsubsection
on p.~\pageref{p:D}).
Take $S$ so large that the probability that a Brownian motion started from $1$ at time 0
is positive over the time interval $[0,S]$ is less than $\epsilon$.

Let $\mathfrak{K}\subseteq C_1[0,S]$ be a compact set
whose Wiener measure
(the distribution of a Brownian motion $W^1$ on $C[0,S]$ starting from 1)
is more than $1-\epsilon$.
Let $f$ be the optimal modulus of continuity for all $\psi\in\mathfrak{K}$:
\[
  f(\delta)
  :=
  \sup_{\substack{(t_1,t_2)\in[0,S]^2:\left|t_1-t_2\right|\le\delta,\\\psi\in\mathfrak{K}}}
  \lvert\psi(t_1)-\psi(t_2)\rvert,
  \quad
  \delta>0;
\]
$f$ is an increasing function, $f(a+b)\le f(a)+f(b)$ for all $a,b\in[0,\infty)$,
and we know that $\lim_{\delta\to0}f(\delta)=0$
(cf.\ the Arzel\`a--Ascoli theorem).
Extend $\mathfrak{K}$ by including in it all $\omega\in C_1[0,S]$
with $f$ as a modulus of continuity;
$\mathfrak{K}$ will stay compact with $W^1(\mathfrak{K})>1-\epsilon$.
Let $A_2:=\{\omega\in\Omega\st{\tilde\omega}|_{[0,S]}\notin\mathfrak{K}\}$,
where ${\tilde\omega}|_{[0,S]}(t):={\tilde\omega}(D(\omega))$ for $t$ such that $D(\omega)\le t\le S$.

Set $B:=1+f(S)$;
notice that $\sup\omega\le B$ for all $\omega\in\Omega\setminus(A_1\cup A_2)$.

Define $D^{S,f}_N\subseteq[0,B]^{N}\times[0,1]^{N}$
to be the set of all sequences
\[(x_1,\ldots,x_N;v_1,\ldots,v_N)\in[0,B]^{N}\times[0,1]^{N}\]
satisfying
\begin{equation}\label{eq:D}
  \begin{cases}
    v_0:=0\le v_1\le\cdots\le v_N\le v_{N+1}:=1,\\
    \left|x_{j}-x_{i}\right|\le f((j-i)S/N)
    \text{ for all $i,j\in\{0,\ldots,N\}$ such that $i<j$},
  \end{cases}
\end{equation}
where $x_0:=1$
(notice that we do not require $v_i<v_{i+1}$ when $v_{i+1}<1$,
in order to make the set \eqref{eq:D} closed).
Define a function $U^{S,f}_N:D^{S,f}_N\to[0,\sup F]$ by
\begin{equation}\label{eq:U}
  U^{S,f}_N(x_1,\ldots,x_N;v_1,\ldots,v_N)
  :=
  F'
  \left(
    A^{S,f}_N(x_1,\ldots,x_N;v_1,\ldots,v_N)
  \right),
\end{equation}
where $F'$ is the continuous function on $\Omega'$ defined earlier
and the set $A:=A^{S,f}_N(x_1,\ldots,x_N;v_1,\ldots,v_N)\in\Omega'$ is defined by the following conditions:
\begin{itemize}
\item
  for all $i\in\{0,\ldots,N\}$ and $t\in(v_i,v_{i+1})$,
  \[
    A^t
    =
    [x_i\wedge x_{i+1}-f(S/N),x_i\vee x_{i+1}+f(S/N)]
  \]
  (with $x_i\wedge x_{i+1}=x_i\vee x_{i+1}:=x_N$ when $i=N$);
\item
  for all $i,j\in\{0,\ldots,N\}$ such that $i<j$ and $v_{i}<t:=v_{i+1}=v_{i+2}=\cdots=v_{j}<v_{j+1}$,
  \[
    A^t
    =
    \left[
      \bigwedge_{k=i}^{j+1} x_{k}-f(S/N),
      \bigvee_{k=i}^{j+1} x_{k}+f(S/N)
    \right];
  \]
\item
  $
    A^0
    =
    [1\wedge x_{1}-f(S/N),1\vee x_{1}+f(S/N)]
  $;
\item
  $A^1=[0,\infty)$.
\end{itemize}
Therefore, $A$ consists of a sequence of horizontal slabs of width at least $2f(S/N)$
separated by vertical lines.
This set contains $\{1\}\times[0,\infty)$
and, for all $i=0,\ldots,N$, also contains $(v_i,x_i)$.

The metric on $D^{S,f}_N$ is defined by
\begin{multline}\label{eq:metric}
  \rho
  \left(
    (x_1,\ldots,x_N;v_1,\ldots,v_N),
    (x'_1,\ldots,x'_N;v'_1,\ldots,v'_N)
  \right)\\
  :=
  \bigvee_{j=1}^N
  \rho_H
  \left(
    (v_j,x_j),(v'_j,x'_j)
  \right),
\end{multline}
where the metric $\rho_H$ on $[0,1]\times[0,\infty)$ is defined by
\begin{align}
  \rho_H
  \left(
    (v,x),(v',x')
  \right)
  &:=
  H
  \left(
    \{(v,x)\} \cup \{1\}\times[0,\infty),
    \{(v',x')\} \cup \{1\}\times[0,\infty)
  \right)\notag\\
  &=
  \left(
    \left|v-v'\right|
    \vee
    \left|x-x'\right|
  \right)
  \wedge
  (1-v\wedge v'),
  \label{eq:rho-H}
\end{align}
$H$ standing for the Hausdorff metric defined in terms of the $\ell_{\infty}$ metric on $[0,1]\times[0,\infty)$,
as before.

\begin{lemma}
  Each function $U^{S,f}_N$ is continuous on $D^{S,f}_N$
  under our definition \eqref{eq:U}
  and the metric \eqref{eq:metric}.
\end{lemma}

\begin{proof}
  Fix some $(x_1,\ldots,x_N;v_1,\ldots,v_N)\in D^{S,f}_N$.
  Let $(x_1^n,\ldots,x^n_N;v^n_1,\ldots,v^n_N)\in D^{S,f}_N$ for $n=1,2,\ldots$
  and $(x_1^n,\ldots,x^n_N;v^n_1,\ldots,v^n_N)\to(x_1,\ldots,x_N;v_1,\ldots,v_N)$
  in $\rho$ as $n\to\infty$.
  It is easy to see that, in the Hausdorff metric,
  \begin{equation*}
    A^{S,f}_N(x_1^n,\ldots,x_N^n,v_1^n,\ldots,v_N^n)
    \to
    A^{S,f}_N(x_1,\ldots,x_N;v_1,\ldots,v_N)
  \end{equation*}
  as $n\to\infty$.
  This implies
  \[
    U^{S,f}_N(x_1^n,\ldots,x_N^n;v_1^n,\ldots,v_N^n)
    \to
    U^{S,f}_N(x_1,\ldots,x_N;v_1,\ldots,v_N)
  \]
  as $n\to\infty$ and completes the proof.
\end{proof}

Define a functional $F^{S,f}_N:\Omega\to[0,\sup F]$ by
\begin{multline}\label{eq:F}
  F^{S,f}_N(\omega)
  =
  U^{S,f}_N
  \bigl(
    \omega(\phi_{S/N}(\omega)\wedge1),\omega(\phi_{2S/N}(\omega)\wedge1),\ldots,\omega(\phi_{S}(\omega)\wedge1);\\
    \phi_{S/N}(\omega)\wedge1,\phi_{2S/N}(\omega)\wedge1,\ldots,\phi_{S}(\omega)\wedge1
  \bigr),
  \qquad
  \omega\in\Omega;
\end{multline}
when the argument on the right-hand side is outside the domain $D^{S,f}_N$ of $U^{S,f}_N$,
set $F^{S,f}_N(\omega):=\sup F$.

The following lemma lists the main properties of the sequence of functionals $F^{S,f}_N$, $N=1,2,\ldots$,
that we will need.
\begin{lemma}\label{lem:main-lemma}
  For all $\omega\in\Omega\setminus(A_1\cup A_2)$,
  \begin{equation*}
    \forall N:
    F^{S,f}_N(\omega) \le F(\omega)
  \end{equation*}
  and
  \begin{equation}\label{eq:limit}
    \liminf_{N\to\infty}
    F^{S,f}_N(\omega)
    \ge
    F(\omega).
  \end{equation}
\end{lemma}
\begin{proof}
  Notice that $\omega\notin A_1$ implies $\phi_S(\omega)=1$;
  therefore, $\omega\in\Omega\setminus(A_1\cup A_2)$ implies
  \begin{multline*}
    \bar\omega
    \subseteq
    A^{S,f}_N
    \bigl(
      \omega(\phi_{S/N}(\omega)\wedge1),\omega(\phi_{2S/N}(\omega)\wedge1),\ldots,\omega(\phi_{S}(\omega)\wedge1);\\
      \phi_{S/N}(\omega)\wedge1,\phi_{2S/N}(\omega)\wedge1,\ldots,\phi_{S}(\omega)\wedge1
    \bigr),
  \end{multline*}
  which immediately implies $F^{S,f}_N(\omega)\le F(\omega)$.
  Since for $\omega\in\Omega\setminus(A_1\cup A_2)$ the Hausdorff distance between
  \begin{multline*}
    A^{S,f}_N
    \bigl(
      \omega(\phi_{S/N}(\omega)\wedge1),\omega(\phi_{2S/N}(\omega)\wedge1),\ldots,\omega(\phi_{S}(\omega)\wedge1);\\
      \phi_{S/N}(\omega)\wedge1,\phi_{2S/N}(\omega)\wedge1,\ldots,\phi_{S}(\omega)\wedge1
    \bigr)
  \end{multline*}
  and $\bar\omega$ tends to 0 as $N\to\infty$,
  we also have \eqref{eq:limit}.
\end{proof}

Let us extend $U^{S,f}_N$ to the whole of
\begin{equation*}
  \left\{
    (x_1,\ldots,x_N;v_1,\ldots,v_N)\in[0,B]^{N}\times[0,1]^{N}
    \st
    v_1\le\cdots\le v_N
  \right\}
\end{equation*}
obtaining a continuous function $\tilde U_N$ taking values in $[0,\sup F]$;
this is possible by the Tietze--Urysohn theorem
(see, e.g., \cite{Engelking:1989}, 2.1.8).
Since the domain of the function $\tilde U_N$ is compact
(in the usual topology, let alone in the topology generated by $\rho$),
this function is uniformly continuous.
Finally, extend $\tilde U_N$ to the whole of
\begin{equation*}
  D_N
  :=
  \left\{
    (x_1,\ldots,x_N;v_1,\ldots,v_N)\in[0,\infty)^{N}\times[0,1]^{N}
    \st
    v_1\le\cdots\le v_N
  \right\}
\end{equation*}
by
\[
  U_N(x_1,\ldots,x_N;v_1,\ldots,v_N)
  :=
  \tilde U_N(x_1\wedge B,\ldots,x_N\wedge B;v_1,\ldots,v_N).
\]
The function $U_N$ inherits the uniform continuity of $\tilde U_N$.

Analogously to \eqref{eq:F},
define a functional $F_N$ by
\begin{multline}\label{eq:F_N}
  F_N(\omega)
  =
  U_N
  \bigl(
    \omega(\phi_{S/N}(\omega)\wedge1),\omega(\phi_{2S/N}(\omega)\wedge1),\ldots,\omega(\phi_{S}(\omega)\wedge1);\\
    \phi_{S/N}(\omega)\wedge1,\phi_{2S/N}(\omega)\wedge1,\ldots,\phi_{S}(\omega)\wedge1
  \bigr);
\end{multline}
by the definition of $U_N$, $F_N(\omega)=F^{S,f}_N(\omega)$ when $\omega\in\Omega\setminus(A_1\cup A_2)$.

Our task is now reduced to proving $\UEG(F_N)\le\UEM(F_N)$.
To demonstrate this,
we first notice that
\begin{align}
  \UPG(A_1) &\le \epsilon, & \UPG(A_2) &\le \epsilon,
    \label{eq:c}\\
  \UPM(A_1) &\le\UPG(A_1) \le \epsilon, & \UPM(A_2) &\le\UPG(A_2) \le \epsilon;
    \label{eq:d}
\end{align}
indeed, \eqref{eq:c} follows from Theorem~3.1 of \cite{\CTIV}
and the time-superinvariance of the sets
\[
  \left\{
    \omega\in C_1[0,\infty)
    \st
    \text{$\tilde\omega$ is defined and positive over $[0,S]$}
  \right\}
\]
and
\[
  \left\{
    \omega\in C_1[0,\infty)
    \st
    \text{$\tilde\omega$ is defined over $[0,S]$ and $\tilde\omega|_{[0,S]}\notin\mathfrak{K}$}
  \right\},
\]
and \eqref{eq:d} follows from $\UPM\le\UPG$, established in the previous subsection:
see \eqref{eq:m-le-g}.
In combination with Lemmas~\ref{lem:outer}, \ref{lem:Fatou}, \ref{lem:main-lemma},
and the assumption $\UEG(F_N)\le\UEM(F_N)$, for all $N$,
this implies
\begin{align*}
  \UEG(F)
  &\le
  \UEG
  \left(
    \liminf_{N\to\infty} F^{S,f}_N
  \right)
  +
  2C\epsilon
  \le
  \liminf_{N\to\infty}
  \UEG(F^{S,f}_N)
  +
  2C\epsilon\\
  &\le
  \liminf_{N\to\infty}
  \UEG(F_N)
  +
  4C\epsilon
  \le
  \liminf_{N\to\infty}
  \UEM(F_N)
  +
  4C\epsilon\\
  &\le
  \liminf_{N\to\infty}
  \UEM(F^{S,f}_N)
  +
  6C\epsilon
  \le
  \UEM(F)
  +
  8C\epsilon
\end{align*}
for $C:=\sup F$.
Since $\epsilon$ can be arbitrarily small, this achieves our goal.

\subsubsection{Setting intermediate goals}

Let us fix $S$ and $N$;
our goal is to prove $\UEG(F_N)\le\UEM(F_N)$.
We will abbreviate $U_N$ to $U$.

We start the proof by defining functions
\begin{align*}
  U^{\rm e}_i:
  D^{\rm e}_i \to [0,\infty),
  &\quad
  i=0,\ldots,N,\\
  U^{\rm m}_i:
  D^{\rm m}_i \to [0,\infty),
  &\quad
  i=0,\ldots,N-1
\end{align*}
(with ``m'' standing for ``maximization'' and ``e'' for ``expectation'')
whose domains are
\begin{align*}
  D^{\rm e}_i
  &:=
  \Bigl\{
    (x_1,v_1,\ldots,x_i,v_i) \in ([0,\infty)\times[0,1])^i
    \st{}\\
  &\qquad
    v_1\le\cdots\le v_i \text{ and }
    (x_j=x_{j+1} \text{ whenever $j<i$ and $v_j=1$})
  \Bigr\},\\
  D^{\rm m}_i
  &:=
  \bigl\{
    (x_1,v_1,\ldots,x_i,v_i,x_{i+1}) \in ([0,\infty)\times[0,1])^i\times[0,\infty)
    \st{}\\
  &\qquad
    v_1\le\cdots\le v_i \text{ and }
    (x_j=x_{j+1} \text{ whenever $j\le i$ and $v_j=1$})
  \bigr\}.
\end{align*}
They will be defined by induction in $i$.

The basis of induction is
\begin{equation}\label{eq:basis-U}
  U^{\rm e}_N(x_1,v_1,\ldots,x_N,v_N)
  :=
  U(x_1,\ldots,x_{N};v_1,\ldots,v_{N}).
\end{equation}
Given $U^{\rm e}_{i+1}$, where $i:=N-1$,
we define
\begin{equation}\label{eq:m-U}
  U^{\rm m}_i(x_1,v_1,\ldots,x_i,v_i,x_{i+1})
  :=
  \sup_{v\in[v_i,1]} U^{\rm e}_{i+1}(x_1,v_1,\ldots,x_i,v_i,x_{i+1},v).
\end{equation}
Given $U^{\rm m}_i$, where $i:=N-1$,
we next define
\begin{equation}\label{eq:e-U}
  U^{\rm e}_i(x_1,v_1,\ldots,x_i,v_i)
  =
  \begin{cases}
    U^{\rm m}_i(x_1,v_1,\ldots,x_i,v_i,x_i) & \text{if $v_i=1$}\\
    \Expect U^{\rm m}_i(x_1,v_1,\ldots,x_i,v_i,\xi) & \text{otherwise}
  \end{cases}
\end{equation}
where $\xi\ge0$ is the value at time $S/N$ of a linear Brownian motion that starts at $x_i$ at time $0$
and is stopped when it hits level $0$.
Next use alternately \eqref{eq:m-U} and \eqref{eq:e-U} for
\[
  i=N-2,N-2;N-3,N-3;\ldots;1,1
\]
to define inductively other $U^{\rm m}_i$ and $U^{\rm e}_i$.
Finally, define
\[
  U^{\rm m}_0(x_1)
  :=
  \sup_{v\in[0,1]} U^{\rm e}_1(x_1,v),
  \qquad
  U^{\rm e}_0
  :=
  \Expect U^{\rm m}_0(\xi)
\]
where $\xi\ge0$ is the value at time $S/N$ of a linear Brownian motion that starts at $1$ at time $0$
and is stopped when it hits level $0$
(the last event being unlikely for a large $N$).

In this proof we will show that $U^{\rm e}_0$ is sandwiched
between $\UEM(F_N)$ and $\UEG(F_N)$
as $\UEM(F_N)\ge U^{\rm e}_0\ge\UEG(F_N)$,
which will achieve our goal.
But first we discuss some properties of regularity
of the intermediate functions $U^{\rm m}_i$ and $U^{\rm e}_i$.

It is obvious that each of the functions $U^{\rm e}_i$ and $U^{\rm m}_i$ is bounded
(by $\sup F$),
and the following two lemmas imply that they are uniformly continuous.
The metric on $D^{\rm e}_i$ is defined by
\begin{equation*}
  \rho^{\rm e}
  \left(
    (x_1,v_1,\ldots,x_i,v_i),
    (x'_1,v'_1,\ldots,x'_i,v'_i)
  \right)
  :=
  \bigvee_{j=1}^i
  \rho_H
  \left(
    (v_{j},x_j),
    (v'_{j},x'_j)
  \right),
\end{equation*}
$\rho_H$ being defined in \eqref{eq:rho-H}.
The metric on $D^{\rm m}_i$ is defined by
\begin{multline*}
  \rho^{\rm m}
  \left(
    (x_1,v_1,\ldots,x_i,v_i,x_{i+1}),
    (x'_1,v'_1,\ldots,x'_i,v'_i,x'_{i+1})
  \right)\\
  :=
  \sup_{v\in[v_i\wedge v'_i,1]}
  \rho^{\rm e}
  \Bigl(
    (x_1,v_1,\ldots,x_i,v_i,x_{i+1},v\vee v_i),\\[-3mm]
    (x'_1,v'_1,\ldots,x'_i,v'_i,x'_{i+1},v\vee v'_i)
  \Bigr).
\end{multline*}

\begin{lemma}
  If a function $U^{\rm e}_{i+1}$ on $D^{\rm e}_{i+1}$ is uniformly continuous,
  then the function $U^{\rm m}_i$ on $D^{\rm m}_i$ defined by \eqref{eq:m-U}
  is also uniformly continuous (with the same modulus of continuity).
\end{lemma}

\begin{proof}
  Let $f$ be a modulus of continuity for $U^{\rm e}_{i+1}$
  (in this paper we only consider increasing moduli of continuity).
  It suffices to prove that, for each $\delta>0$,
  \begin{multline}\label{eq:to-prove-1}
    \sup_{v\in[v_i,1]} U^{\rm e}_{i+1}(x_1,v_1,\ldots,x_i,v_i,x_{i+1},v)\\
    \ge
    \sup_{v\in[v'_i,1]} U^{\rm e}_{i+1}(x'_1,v'_1,\ldots,x'_i,v'_i,x'_{i+1},v)
    -
    f(\delta)
  \end{multline}
  provided the $D^{\rm m}_i$ distance between $(x_1,v_1,\ldots,x_{i+1})$ and $(x'_1,v'_1,\ldots,x'_{i+1})$
  does not exceed $\delta$.
  This follows from
  \begin{multline*}
    \sup_{v\in[v_i,1]} U^{\rm e}_{i+1}(x_1,v_1,\ldots,x_i,v_i,x_{i+1},v)
    \ge
    U^{\rm e}_{i+1}(x_1,v_1,\ldots,x_i,v_i,x_{i+1},v'\vee v_i)\\
    \ge
    U^{\rm e}_{i+1}(x'_1,v'_1,\ldots,x'_i,v'_i,x'_{i+1},v')
    -
    f(\delta)\\
    =
    \sup_{v\in[v'_i,1]} U^{\rm e}_{i+1}(x'_1,v'_1,\ldots,x'_i,v'_i,x'_{i+1},v)
    -
    f(\delta),
  \end{multline*}
  where $v'\ge v'_i$ is the point at which the supremum on the right-hand side of \eqref{eq:to-prove-1} is attained.
\end{proof}

\begin{lemma}\label{lem:uniform-continuity-e}
  If a function $U^{\rm m}_i$ on $D^{\rm m}_i$ is bounded and uniformly continuous,
  then the function $U^{\rm e}_i$ on $D^{\rm e}_i$ defined by \eqref{eq:e-U}
  is also uniformly continuous.
\end{lemma}
\begin{proof}
  Let $\delta>0$ and $f$ be the optimal modulus of continuity for $U^{\rm m}_i$;
  to bound the optimal modulus of continuity for $U^{\rm e}_i$,
  we consider three possibilities for two points $E$ and $E'$ in $D^{\rm e}_i$
  where
  $E=(x_1,v_1,\ldots,x_i,v_i)$,
  $E'=(x'_1,v'_1,\ldots,x'_i,v'_i)$,
  and $\rho^{\rm e}(E,E')\le\delta$.
  \begin{itemize}
  \item
    If $v_i=v'_i=1$,
    the difference between $U^{\rm e}_i(E)$ and $U^{\rm e}_i(E')$
    does not exceed $f(\delta)$.
  \item
    If $v_i<v'_i=1$ or $v'_i<v_i=1$,
    the difference between $U^{\rm e}_i(E)$ and $U^{\rm e}_i(E')$
    also does not exceed $f(\delta)$.
    Indeed, suppose, for concreteness, that $v'_i=1$.
    Then $v_i\ge1-\delta$.
    By definition, $U^{\rm e}_i(E)$ is an average of $U^{\rm m}_i(E,x)$ over $x$,
    and $U^{\rm e}_i(E')$ coincides with $U^{\rm m}_i(E',x'_i)$.
    By the definition of the metric on $D^{\rm m}_i$,
    the $\rho^{\rm m}$ distance between $(E,x)$ and $(E',x'_i)$ is at most $\delta$,
    and so the difference between $U^{\rm e}_i(E)$ and $U^{\rm e}_i(E')$
    does not exceed $f(\delta)$.
  \item
    If $v_i<1$ and $v'_i<1$,
    the difference between $U^{\rm e}_i(E)$ and $U^{\rm e}_i(E')$
    does not exceed $2f(\delta)+C\delta\sqrt{N/S}$,
    where $C$ is an upper bound on $U^{\rm m}_i$.
    Let us check this.
    Our goal is to prove that
    \[
      \left|
        \Expect U^{\rm m}_i(E,\xi)
        -
        \Expect U^{\rm m}_i(E',\xi')
      \right|
      \le
      2f(\delta)+C\delta\sqrt{N/S}
    \]
    where $\xi$ (resp.\ $\xi'$) is the value at time $S/N$ of a linear Brownian motion
    that starts at $x_i$ (resp.\ $x'_i$) at time $0$
    and is stopped when it hits level $0$.
    It suffices to notice that
    \begin{align}
      &\left|
        \Expect U^{\rm m}_i(E,\xi)
        -
        \Expect U^{\rm m}_i(E',\xi')
      \right|\notag\\
      &\le
      \left|
        \Expect U^{\rm m}_i(E,\xi)
        -
        \Expect U^{\rm m}_i(E,\xi')
      \right|
      +
      \left|
        \Expect U^{\rm m}_i(E,\xi')
        -
        \Expect U^{\rm m}_i(E',\xi')
      \right|\label{eq:addends}\\
      &\le
      f(\delta)+C\delta/\sqrt{S/N}
      +
      f(\delta).\notag
    \end{align}
    The upper bound $f(\delta)+C\delta/\sqrt{S/N}$
    on the first addend in \eqref{eq:addends}
    follows from Lemma~\ref{lem:auxiliary} below;
    we also used the uniform continuity of $U^{\rm m}_i(E,\cdot)$ and $U^{\rm m}_i(\cdot,x)$,
    where $x\in[0,\infty)$, with $f$ as modulus of continuity.
  \end{itemize}
  In all three cases the difference is bounded by $2f(\delta)+C\delta\sqrt{N/S}$.
\end{proof}

The following result was used in the proof of Lemma~\ref{lem:uniform-continuity-e} above.
\begin{lemma}\label{lem:auxiliary}
  Suppose $a>0$ and $u:[0,\infty)\to[0,C]$ is a bounded uniformly continuous function
  with $f$ as modulus of continuity.
  Then
  \[
    x\in[0,\infty)
    \mapsto
    \Expect u(W^x_{\tau\wedge a}),
  \]
  where $W^x$ is a Brownian motion started at $x$
  and $\tau$ is the moment it hits level $0$,
  is uniformly continuous with
  $\delta>0\mapsto f(\delta)+C\delta/\sqrt{a}$
  as modulus of continuity.
\end{lemma}
\begin{proof}
  Consider points $x\in[0,\infty)$ and $x'\in(x,x+\delta]$, for some $\delta>0$.
  Let us map each path of $W^{x'}_{\tau\wedge a}$ to the path of $W^x_{\tau\wedge a}$
  obtained by subtracting $x'-x$ and stopping when level 0 is hit;
  we will refer to the latter as the path \emph{corresponding} to the former.
  There are three kinds of paths of $W^{x'}_{\tau\wedge a}$:
  \begin{itemize}
  \item
    Those that never hit level $x'-x$ over the time interval $[0,a]$.
    The average of $u(W^{x'}_{\tau\wedge a})=u(W^{x'}_{a})$
    over such paths
    and the average of $u(W^{x}_{\tau\wedge a})=u(W^{x}_{a})$
    over the corresponding paths
    differ by at most $f(\delta)$.
  \item
    Those that hit level $0$ over $[0,a]$.
    The average of $u(W^{x'}_{\tau\wedge a})=u(0)$
    over such paths
    and the average of $u(W^{x}_{\tau\wedge a})=u(0)$
    over the corresponding paths
    coincide.
  \item
    Those that hit level $x'-x$ but never hit level $0$ over $[0,a]$.
    The probability of such paths is
    \begin{align*}
      2\Phi(-x/\sqrt{a})-2\Phi(-x'/\sqrt{a})
      &\le
      2\Prob(\xi\in[0,(x'-x)/\sqrt{a}])\\
      &<
      \frac{2}{\sqrt{2\pi}}
      (x'-x)/\sqrt{a}
      <
      \delta/\sqrt{a},
    \end{align*}
    where $\Phi$ is the standard normal distribution function, $\xi\sim\Phi$,
    and the factor of 2 comes from the reflection principle.
  \end{itemize}
  Therefore, the overall averages of $u(W^{x}_{\tau\wedge a})$ and $u(W^{x'}_{\tau\wedge a})$
  differ by at most $f(\delta)+C\delta/\sqrt{a}$.
\end{proof}

\subsubsection{Tackling measure-theoretic probability}

First we prove an easy auxiliary statement ensuring the existence of measurable ``choice functions''.
\begin{lemma}\label{lem:choice}
  Suppose $\{A_{\theta}\st\theta\in\Theta\}$ is a countable cover of a measurable space $\Omega$
  such that each $A_{\theta}$ is measurable.
  There is a measurable function $f:\Omega\to\Theta$
  (with the discrete $\sigma$-algebra on $\Theta$)
  such that $\omega\in A_{f(\omega)}$ for all $\omega\in\Omega$.
\end{lemma}
\begin{proof}
  Assume, without loss of generality, $\Theta=\mathbb{N}$.
  Define
  \[
    f(\omega):=\min\{\theta\st\omega\in A_{\theta}\}.
  \]
  Then, for each $\theta\in\mathbb{N}$, the set
  \[
    \{\omega\st f(\omega)\le\theta\}
    =
    A_1 \cup \cdots \cup A_{\theta}
  \]
  is measurable.
\end{proof}

In this section we show that $\UEM(F_N)\ge U_0^{\rm e}$.
We define a martingale measure $P$ by backward induction.
For each $i=0,\ldots,N-1$, let $V_{i+1}$ be a Borel function on $D^{\rm m}_i$ such that,
for all $(x_1,v_1,\ldots,x_i,v_i,x_{i+1})\in D^{\rm m}_i$ satisfying $v_i<1$,
it is true that
\[
  v_i < V_{i+1}(x_1,v_1,\ldots,x_i,v_i,x_{i+1}) < 1
\]
and
\begin{multline*}
  U^{\rm e}_{i+1}
  \left(
    x_1,v_1,\ldots,x_i,v_i,x_{i+1},V_{i+1}(x_1,v_1,\ldots,x_i,v_i,x_{i+1})
  \right)\\
  \ge
  U^{\rm m}_i(x_1,v_1,\ldots,x_i,v_i,x_{i+1})
  -
  \epsilon
\end{multline*}
(cf.\ \eqref{eq:m-U}),
where $\epsilon>0$ is a small constant (further details will be added later).
(Intuitively, $V_{i+1}$ outputs a $v>v_i$
at which the supremum of $U^{\rm e}_{i+1}(x_1,v_1,\ldots,x_{i+1},v)$ is almost attained.)
The existence of such $V_{i+1}$ follows from Lemma~\ref{lem:choice}:
indeed, for each rational $r\in(0,1)$ the set
\begin{multline*}
  A_r
  :=
  \bigl\{
    (x_1,v_1,\ldots,x_i,v_i,x_{i+1})\in D^{\rm m}_i
    \st
    r>v_i
    \text{ and }\\
    U^{\rm e}_{i+1}(x_1,v_1,\ldots,x_i,v_i,x_{i+1},r)
    \ge
    U^{\rm m}_{i}(x_1,v_1,\ldots,x_i,v_i,x_{i+1})
    -
    \epsilon
  \bigr\}
\end{multline*}
is Borel (namely, intersection of open and closed),
and the sets $A_r$ form a cover of $D^{\rm m}_i$.
By the uniform continuity of $U^{\rm e}_{i+1}$ and $U^{\rm m}_{i}$,
there is $\delta>0$ such that,
for all $i$ (remember that there are finitely many $i$)
and for all $x_1$, $v_1$,\ldots, $x_i$, $v_i$, $x_{i+1}$, and $x'_{i+1}$,
\begin{multline}\label{eq:stronger}
  \left|
    x'_{i+1} - x_{i+1}
  \right|
  <
  \delta\\
  \Longrightarrow
  U^{\rm e}_{i+1}
  \left(
    x_1,v_1,\ldots,x_i,v_i,x_{i+1},V_{i+1}(x_1,v_1,\ldots,x_i,v_i,x'_{i+1})
  \right)\\
  \ge
  U^{\rm m}_i(x_1,v_1,\ldots,x_i,v_i,x'_{i+1})
  -
  2\epsilon.
\end{multline}
Next choose Borel $V^*_i$ such that, for $v_i<1$,
\begin{equation}\label{eq:prob-epsilon-1}
  V_{i+1}(x_1,v_1,\ldots,x_i,v_i,\xi)
  >
  V^*_i(x_1,v_1,\ldots,x_i,v_i)
  >
  v_i
\end{equation}
with probability (over $\xi$ only) at least $1-\epsilon$
when $\xi$ is the value taken at time $S/N$
by a linear Brownian motion started from $x_i$ at time $0$ and stopped when it hits level $0$.
(The existence of $V^*_i$ also follows from Lemma~\ref{lem:choice}.)
Let $\Delta\in(0,S/N)$ be such that
\begin{equation}\label{eq:prob-epsilon-2}
  \sup_{t\in[0,\Delta]}
  \lvert W_t\rvert
  <
  \delta
\end{equation}
with a probability at least $1-\epsilon$,
where $W$ is a standard Brownian motion.

By a \emph{scaled Brownian motion} we will mean a process of the type $W_{ct}$
where $W$ is a Brownian motion and $c>0$
(equivalently, a process of the type $cW_t$ where $W$ is a Brownian motion and $c>0$).
Define a probability measure $P$ on $\Omega$ as the distribution of $\omega\in\Omega$ generated as follows.
For $i=0,1,\ldots,N-1$:
\begin{itemize}
\item
  Start a scaled Brownian motion $W^i$ (independent of what has happened before if $i>0$)
  from $x_i$ (with $x_0:=1$) at time $v_i$ (with $v_0:=0$)
  such that its quadratic variation over $[v_i,v^*_i]$ is $S/N-\Delta$, where
  \[
    v^*_i := V^*_i(x_1,v_1,\ldots,x_i,v_i) < 1.
  \]
  Define
  \begin{equation*}
    \omega|_{[v_i,v^*_i]}
    :=
    W^{\circ,i}|_{[v_i,v^*_i]}
  \end{equation*}
  where $W^{\circ,i}$ is $W^i$ stopped when it hits level 0.
  If $\omega(v^*_i)=0$, the random process of generating $\omega$ is complete;
  set $\omega|_{[v_i^*,1]}:=0$, $v_{i+1}^*=\cdots=v_{N-1}^*:=1$, and $v_{i+1}=\cdots=v_{N}:=1$,
  and then stop.
\item
  Set
  \begin{multline*}
    v_{i+1}
    :=\\
    \begin{cases}
      V_{i+1}(x_1,v_1,\ldots,x_i,v_i,\omega(v^*_i))
        & \text{if $V_{i+1}(x_1,v_1,\ldots,x_i,v_i,\omega(v^*_i))>v^*_i$}\\
      1 & \text{otherwise}.
    \end{cases}
  \end{multline*}
  Start another independent Brownian motion $\bar W^i$ from $\omega(v^*_i)$ at time $v^*_i$
  such that its quadratic variation over $[v^*_i,v_{i+1}]$ is $\Delta$.
  Define
  \[
    \omega|_{[v^*_i,v_{i+1}]}
    :=
    \bar W^{\circ,i}|_{[v^*_i,v_{i+1}]}
  \]
  where $\bar W^{\circ,i}$ is $\bar W^i$ stopped when it hits level 0.
  If $\omega(v_{i+1})=0$ or $v_{i+1}=1$ (or both),
  the random process of generating $\omega$ is complete;
  set $\omega|_{[v_{i+1},1]}:=0$ if $v_{i+1}<1$,
  set $v_{i+1}^*=\cdots=v_{N-1}^*:=1$ and $v_{i+2}=\cdots=v_{N}:=1$,
  and then stop.
\item
  Set $x_{i+1}:=\omega(v_{i+1})$; notice that $v_{i+1}<1$.
\end{itemize}
If the procedure was not stopped, and so $v_N<1$,
define $\omega|_{[v_N,1]}$ to be the constant $x_N=\omega(v_N)$.

Let us now check that $\Expect_P(F_N)\ge U^{\rm e}_0$.
More precisely, we will show by induction in $i$ that,
for $i=N,\ldots,0$,
\begin{equation}\label{eq:e-goal}
  \Expect_P(F_N\given\FFF_{\tilde v_i})
  \ge
  U^{\rm e}_i
  \left(
    \tilde x_1, \tilde v_1,
    \ldots,
    \tilde x_i, \tilde v_i
  \right)
  -
  (N-i)(3C+3)\epsilon
  \quad
  \text{a.s.},
\end{equation}
and that, for $i=N-1,\ldots,0$,
\begin{multline}\label{eq:m-goal}
  \Expect_P(F_N\given\FFF_{\tilde v^*_i})
  \ge
  U^{\rm m}_i
  \left(
    \tilde x_1, \tilde v_1,
    \ldots,
    \tilde x_i, \tilde v_i,
    \omega(\tilde v^*_i)
  \right)\\
  -
  (N-i)(3C+3)\epsilon
  +
  (C+1)\epsilon
  \quad
  \text{a.s.},
\end{multline}
where: $C:=\sup U$;
$\tilde x_j$ are $x_j$ (as defined in the definition of $P$)
considered as function of $\omega$
(it is clear that $x_j$ can be restored given $\omega$ $P$-almost surely);
similarly, $\tilde v_j$ and $\tilde v^*_j$ are $v_j$ and $v^*_j$
considered as functions of $\omega$;
$\FFF_{\tilde v_i}$ and $\FFF_{\tilde v^*_i}$ are the usual $\sigma$-algebras on $\Omega$
defined as in \eqref{eq:tau} for the stopping times $\tilde v_i$ and $\tilde v^*_i$.
Since, $\epsilon$ can be arbitrarily small,
\eqref{eq:e-goal} with $i=0$ will achieve our goal.

For $i=N$, \eqref{eq:e-goal} holds almost surely
as $U^{\rm e}_i:=U:=U_N$ and $F_N$ is defined by \eqref{eq:F_N}.

Assuming \eqref{eq:e-goal} with $i+1$ in place of $i$, $i<N$,
let us deduce \eqref{eq:m-goal}:
concentrating on the non-trivial case $\tilde v_i<1$,
\begin{align*}
  &\Expect_P(F_N\given\FFF_{\tilde v_{i}^*})
  =
  \Expect_P
  \Bigl(
    \Expect_P(F_N\given\FFF_{\tilde v_{i+1}})
    \given
    \FFF_{\tilde v^*_{i}}
  \Bigr)\\
  &\ge
  \Expect_P
  \Bigl(
    U^{\rm e}_{i+1}
    \left(
      \tilde x_1, \tilde v_1,
      \ldots,
      \tilde x_i, \tilde v_i,
      \tilde x_{i+1}, \tilde v_{i+1}
    \right)
    \given
    \FFF_{\tilde v^*_{i}}
  \Bigr)
  -
  (N-i-1)(3C+3)\epsilon\\
  &\ge
  U^{\rm m}_{i}
  \left(
    \tilde x_1, \tilde v_1,
    \ldots,
    \tilde x_i, \tilde v_i,
    \omega(\tilde v^*_i)
  \right)
  -
  (N-i-1)(3C+3)\epsilon
  -
  (2C+2)\epsilon\\
  &=
  U^{\rm m}_{i}
  \left(
    \tilde x_1, \tilde v_1,
    \ldots,
    \tilde x_i, \tilde v_i,
    \omega(\tilde v^*_i)
  \right)
  -
  (N-i)(3C+3)\epsilon
  +
  (C+1)\epsilon
  \quad
  \text{a.s.},
\end{align*}
where the second inequality
follows from the fact that
\[
  U^{\rm e}_{i+1}
  \left(
    \tilde x_1, \tilde v_1,
    \ldots,
    \tilde x_i, \tilde v_i,
    \tilde x_{i+1}, \tilde v_{i+1}
  \right)
  \ge
  U^{\rm m}_{i}
  \left(
    \tilde x_1, \tilde v_1,
    \ldots,
    \tilde x_i, \tilde v_i,
    \omega(\tilde v^*_i)
  \right)
  -
  2\epsilon
\]
with $\FFF_{\tilde v^*_{i}}$-conditional probability at least $1-2\epsilon$ a.s.
This fact in turn follows from \eqref{eq:prob-epsilon-1} and \eqref{eq:prob-epsilon-2}
each holding with probability at least $1-\epsilon$ 
(and so the conjunction of $\left|\omega(\tilde v_{i+1})-\omega(\tilde v^*_i)\right|<\delta$
and $\tilde v_{i+1}<1$
holding with $\FFF_{\tilde v^*_{i}}$-conditional probability at least $1-2\epsilon$ a.s.)\
combined with an application of \eqref{eq:stronger}.

Assuming \eqref{eq:m-goal} let us deduce \eqref{eq:e-goal}:
again concentrating on the case $\tilde v_i<1$,
\begin{align*}
  &\Expect_P(F_N\given\FFF_{\tilde v_{i}})
  =
  \Expect_P
  \Bigl(
    \Expect_P(F_N\given\FFF_{\tilde v^*_i})
    \given
    \FFF_{\tilde v_{i}}
  \Bigr)\\
  &\ge
  \Expect_P
  \Bigl(
    U^{\rm m}_i
    \left(
      \tilde x_1, \tilde v_1,
      \ldots,
      \tilde x_i, \tilde v_i,
      \omega(\tilde v^*_i)
    \right)
    \given
    \FFF_{\tilde v_{i}}
  \Bigr)
  -
  (N-i)(3C+3)\epsilon
  +
  (C+1)\epsilon\\
  &=
  \Expect
  U^{\rm m}_{i}
  \left(
    \tilde x_1, \tilde v_1,
    \ldots,
    \tilde x_i, \tilde v_i,
    \xi
  \right)
  -
  (N-i)(3C+3)\epsilon
  +
  (C+1)\epsilon\\
  &\ge
  U^{\rm e}_{i}
  \left(
    \tilde x_1, \tilde v_1,
    \ldots,
    \tilde x_i, \tilde v_i
  \right)
  -
  (N-i)(3C+3)\epsilon
  \quad
  \text{a.s.}
\end{align*}
where $\xi$ is the value at time $S/N-\Delta$
(rather than $S/N$ as in the definition of $U^{\rm e}_i$)
of a linear Brownian motion started at $\tilde x_i$
at time 0 and stopped when it hits level 0,
and $\Expect$ (without a subscript) refers to averaging over $\xi$ only.
The last inequality can be derived as follows:
\begin{itemize}
\item
  Using the time period $[0,S/N-\Delta]$ in place of $[0,S/N]$
  in the definition of $\xi$,
  we make an error (in the value of $\xi$) of at most $\delta$ with probability at least $1-\epsilon$:
  cf.\ \eqref{eq:prob-epsilon-2}.
\item
  This leads to an error of at most $f(\delta)$ with probability at least $1-\epsilon$
  in the expression
  $
    \Expect
    U^{\rm m}_{i}
    \left(
      \tilde x_1, \tilde v_1,
      \ldots,
      \tilde x_i, \tilde v_i,
      \xi
    \right)
  $,
  where $f$ is a modulus of continuity for all $U^{\rm m}_{i}$, $i=0,\ldots,N-1$.
\item
  Without loss of generality assume $f(\delta)\le\epsilon$.
\end{itemize}

\subsubsection{Tackling game-theoretic probability}

Now we show that $\UEG(F_N)\le U^{\rm e}_0$.

Let $\epsilon>0$ be a small positive number (see below for details of how small),
let $L$ be a large positive integer (see below for details of how large depending on $\epsilon$),
and for each $i=N,N-1,\ldots,0$,
define a function
\[
  \overline{U}_i:
  \mathbb{N}_0\times\{0,1,\ldots,L\}
  \times
  D^{\rm e}_i
  \to
  [0,\infty)
\]
by
\begin{equation}\label{eq:overline-U-basis}
  \overline{U}_i(X,L;x_1,v_1,\ldots,x_i,v_i)
  :=
  U^{\rm m}_i(x_1,v_1,\ldots,x_i,v_i,X\sqrt{S/NL})
\end{equation}
and, for $j=L-1,\ldots,1,0$,
\begin{multline}\label{eq:overline-U}
  \overline{U}_i(X,j;x_1,v_1,\ldots,x_i,v_i)
  :=\\
  \frac{
    \overline{U}_i(X-1,j+1;x_1,v_1,\ldots,x_i,v_i)
    +
    \overline{U}_i(X+1,j+1;x_1,v_1,\ldots,x_i,v_i)}
  {2},
\end{multline}
if $X>0$, and
\begin{equation}\label{eq:overline-U-degenerate}
  \overline{U}_i(0,j;x_1,v_1,\ldots,x_i,v_i)
  :=
  \overline{U}_i(0,j+1;x_1,v_1,\ldots,x_i,v_i).
\end{equation}
Equations~\eqref{eq:overline-U-basis}--\eqref{eq:overline-U-degenerate} assume $v_i<1$;
if $v_i=1$, set, e.g.,
\begin{equation*}
  \overline{U}_i(X,j;x_1,v_1,\ldots,x_i,v_i)
  :=
  U^{\rm m}_i(x_1,v_1,\ldots,x_i,v_i,X\sqrt{S/NL})
\end{equation*}
for all $j=0,\ldots,L$
(although the only interesting case for us is $v_i<1-\epsilon$).
We will fix $i\in\{0,1,\ldots,N\}$ for a while.

Let us check that
\begin{equation}\label{eq:always}
  U^{\rm e}_i
  \left(
    x_1,v_1,\ldots,x_i,v_i
  \right)
  \approx
  \overline{U}_i
  \left(
    \lfloor x_i/\sqrt{S/NL}\rfloor,0;x_1,v_1,\ldots,x_i,v_i
  \right),
\end{equation}
assuming $v_i<1$.
This follows from the KMT theorem
(Theorem~1 of Koml\'os, Major, and Tusn\'ady \cite{Komlos/etal:1976}; see also \cite{Komlos/etal:1975});
we will use its following special case (\cite{Chatterjee:2012}, Theorem~1.5).
\begin{KMTtheorem}
  Let $E_1,E_2,\ldots$ be i.i.d.\ symmetric $\pm1$-valued random variables.
  For each $k$, let $S_k := \sum_{i=1}^k E_i$.
  It is possible to construct a version of the sequence $(S_k)_{k\ge0}$
  and a standard Brownian motion $(B_t)_{t\ge0}$ on the same probability space
  such that, for all $n$ and all $x \ge 0$,
  \[
    \Prob
    \left(
      \max_{k\le n}
      \left|S_k - B_k\right|
      \ge
      C_1 \ln n + x
    \right)
    \le
    C_2 e^{-x},
  \]
  where $C_1$ and $C_2$ are absolute constants.
\end{KMTtheorem}
\noindent
(Although for our purpose much simpler results, such as those \cite{Strassen:1967}
based on Skorokhod's representation, would have been sufficient.)
On the left-hand side of~\eqref{eq:always} we have the average of
$\overline{U}^{\rm m}_i(x_1,v_1,\ldots,x_i,v_i,\cdot)$
w.r.\ to the value of a Brownian motion at time $S/N$ stopped when it hits level $0$
and on the right-hand side of~\eqref{eq:always} we have the average of the same function
w.r.\ to the value of a scaled simple random walk at the same time $S/N$
stopped when it hits level 0;
the scaled random walk makes steps of $S/NL$ in time and $\sqrt{S/NL}$ in space;
the Brownian motion and random walk are started from nearby points,
namely $x_i$ and $\lfloor x_i/\sqrt{S/NL}\rfloor\sqrt{S/NL}$.
By the KMT theorem there are coupled versions of the Brownian motion (not stopped)
and the scaled simple random walk (also not stopped)
that differ by at most $\epsilon$ over $[0,S]$
with probability at least $1-\epsilon$,
provided $L$ is large enough.
(For example, we can take $L$ large enough
for $x_i$ and $\lfloor x_i/\sqrt{S/NL}\rfloor\sqrt{S/NL}$
to be $\epsilon/2$-close
and for the precision of the KMT approximation over $[0,S]$
to be $\epsilon/2$ with probability at least $1-\epsilon$.)
The values at time $S/N$ of the stopped Brownian motion and stopped scaled random walk
can differ by more than $\epsilon$
even when their non-stopped counterparts differ by at most $\epsilon$ over $[0,S]$,
but as the argument in Lemma~\ref{lem:auxiliary} shows,
the probability of this is at most
$
  3\epsilon/\sqrt{S/N}
$
(we would have $2\epsilon/\sqrt{S/N}$ if both coupled processes were Brownian motions,
and replacing $2$ by $3$ adjusts for the discreteness of the random walk, for large $L$).
Therefore, the difference between the two sides of~\eqref{eq:always}
does not exceed
\begin{equation}\label{eq:g}
  g(\epsilon)
  :=
  f(\epsilon)
  +
  C 3\epsilon/\sqrt{S/N},
\end{equation}
where $f$ is a modulus of continuity of $\overline{U}^{\rm m}_i$
for all $i=0,\ldots,N-1$ and $C:=\sup U$.

For $i=1,\ldots,N$, set
\begin{align}
  v_{i}&=v_i(\omega):=\phi_{iS/N}(\omega)\wedge1,\label{eq:v}\\
  x_{i}&=x_i(\omega):=\omega(v_{i}).\label{eq:x}
\end{align}
During each non-empty time interval $[v_i(\omega),v_{i+1}(\omega))$
the trader will bet at the stopping times
\begin{align}
  T_{i,0}(\omega)
  &:=
  \inf
  \left\{
    t\ge v_{i}(\omega)
    \st
    \omega(t)/\sqrt{S/NL}\in\mathbb{N}_0
  \right\},
  \notag\\
  T_{i,j}(\omega)
  &:=
  \inf
  \left\{
    t\ge T_{i,j-1}(\omega)
    \st
    \omega(t)/\sqrt{S/NL}\in\mathbb{N}_0,
    \;
    \omega(t)\ne\omega(T_{i,j-1}(\omega))
  \right\},
  \notag\\
  &\qquad\qquad
  j\in\{1,\ldots,L\},
  \notag
\end{align}
such that $T_{i,j}(\omega) < v_{i+1}(\omega)\wedge(1-\epsilon)$;
therefore, we are only interested in the case $j\in\{1,\ldots,J_i\}$ where
\begin{equation*}
  J_i=J_i(\omega)
  :=
  \max
  \bigl\{
    j\in\{0,\ldots,L\}
    \st
    T_{i,j}(\omega) < v_{i+1}(\omega)
  \bigr\}
\end{equation*}
($J_i=L$ being a common case).
Besides, the bet at the times $v_i(\omega)$ will be set to zero
unless $v_i(\omega)=T_{i,0}(\omega)$.
The bets at the times $T_{i,L}(\omega)$ will also be set to zero
unless $T_{i,L}(\omega)=T_{i+1,0}(\omega)$.

For $j=0,\ldots,L$, set
\begin{equation*}
  X_{i,j}
  :=
  \omega(T_{i,j})/\sqrt{S/NL}
  \in
  \mathbb{N}_0.
\end{equation*}
The bet at time $T_{i,j}(\omega)<1-\epsilon$ is $0$ if $X_{i,j}=0$ or $j=L$;
otherwise, it is defined in such a way that the increase of the capital over $[T_{i,j},T_{i,j+1}]$
is typically
\begin{equation}\label{eq:increase}
  \overline{U}_i(X_{i,j+1},j+1;x_1,v_1,\ldots,x_i,v_i)
  -
  \overline{U}_i(X_{i,j},j;x_1,v_1,\ldots,x_i,v_i)
\end{equation}
(this assumes, e.g., $T_{i,j+1}\le v_{i+1}$);
namely, the bet at time $T_{i,j}$ is formally defined as
\begin{equation}\label{eq:bet}
  \frac
  {
    \overline{U}_i(X_{i,j}+1,j+1;x_1,v_1,\ldots,x_i,v_i)
    -
    \overline{U}_i(X_{i,j},j;x_1,v_1,\ldots,x_i,v_i)
  }
  {\sqrt{S/NL}}.
\end{equation}
(When $X_{i,j+1}>X_{i,j}$, the increase is \eqref{eq:increase} by the definition of the bet,
and when $X_{i,j+1}<X_{i,j}$, the increase is \eqref{eq:increase} by the definition of the bet
and the definition \eqref{eq:overline-U}.)

Let us check that this strategy achieves the final value
greater than or close to $F_N(\omega)$
(with high lower game-theoretic probability) starting from $U^{\rm e}_0$.
More generally, we will check that the capital $\K$ of this strategy
(started with $U^{\rm e}_0$)
at time $v_{i}(\omega)$, $i=0,1,\ldots,N$, satisfies
\begin{equation*}
  \K_{v_{i}(\omega)}
  \gtrsim
  U^{\rm e}_i
  \left(
    x_1(\omega),v_1(\omega),\ldots,x_i(\omega),v_i(\omega)
  \right)
\end{equation*}
with lower game-theoretic probability close to~1,
in the notation of~\eqref{eq:v}--\eqref{eq:x}.
More precisely, we will check that, for $i=0,1,\ldots,N$ such that $v_i(\omega)<1-\epsilon$,
\begin{equation}\label{eq:by-induction}
  \K_{v_{i}(\omega)}
  \ge
  U^{\rm e}_i
  \left(
    x_1(\omega),v_1(\omega),\ldots,x_i(\omega),v_i(\omega)
  \right)
  -
  iA
\end{equation}
with lower game-theoretic probability at least $1-2i\epsilon$,
where
\[
  A
  :=
  3f(\epsilon)
  +
  g(\epsilon)
\]
and $g(\epsilon)$ is defined by \eqref{eq:g}.

We use induction in $i$.
Suppose \eqref{eq:by-induction} holds;
our goal is to prove \eqref{eq:by-induction} with $i+1$ in place of $i$.
We have, for $v_{i+1}<1-\epsilon$:
\begin{align}
  &\K_{v_{i+1}}
  \ge
  \K_{T_{i,J_i}}
  -
  f(\epsilon)
  \label{eq:chain-start}\\
  &=
  \K_{T_{i,0}}
  +
  \overline{U}_i(X_{i,J_i},J_i;x_1,v_1,\ldots,x_i,v_i)
  -
  \overline{U}_i(X_{i,0},0;x_1,v_1,\ldots,x_i,v_i)
  \notag\\
  &\qquad-
  f(\epsilon)
  \notag\\
  &\ge
  \K_{v_{i}}
  +
  \overline{U}_i(X_{i,J_i},J_i;x_1,v_1,\ldots,x_i,v_i)
  -
  U^{\rm e}_i(x_1,v_1,\ldots,x_i,v_i)
  \label{eq:-1}\\
  &\qquad-
  f(\epsilon)
  -
  g(\epsilon)
  \notag\\
  &\ge
  \overline{U}_i(X_{i,J_i},J_i;x_1,v_1,\ldots,x_i,v_i)
  -
  iA
  -
  f(\epsilon)
  -
  g(\epsilon)
  \label{eq:0}\\
  &\ge
  \overline{U}_i(X_{i,J_i},L;x_1,v_1,\ldots,x_i,v_i)
  -
  iA
  -
  2f(\epsilon)
  -
  g(\epsilon)
  \label{eq:1}\\
  &=
  U^{\rm m}_i(x_1,v_1,\ldots,x_i,v_i,X_{i,J_i}\sqrt{S/NL})
  -
  iA
  -
  2f(\epsilon)
  -
  g(\epsilon)
  \label{eq:1.5}\\
  &\ge
  U^{\rm m}_i(x_1,v_1,\ldots,x_i,v_i,x_{i+1})
  -
  iA
  -
  3f(\epsilon)
  -
  g(\epsilon)
  \label{eq:2}\\
  &\ge
  U^{\rm e}_i(x_1,v_1,\ldots,x_i,v_i,x_{i+1},v_{i+1})
  -
  iA
  -
  3f(\epsilon)
  -
  g(\epsilon)
  \label{eq:chain-end}
\end{align}
where:
\begin{itemize}
\item
  the inequality~\eqref{eq:chain-start} holds
  for a large enough $L$
  and follows from the form \eqref{eq:bet} of the bets
  (called off at time $T_{i,L}$)
  and the uniform continuity of $U^{\rm m}_i$
  (which propagates to $\overline{U}_i$)
  with $f$ as modulus of continuity (for all $i$);
  the error term $f(\sqrt{S/NL})$ is replaced by the cruder $f(\epsilon)$;
\item
  the inequality~\eqref{eq:-1} follows from the approximate equality \eqref{eq:always},
  whose accuracy is given by \eqref{eq:g}
  (notice that the accuracy \eqref{eq:g} is also applicable to \eqref{eq:always}
  with $\lceil\cdots\rceil$ in place of $\lfloor\cdots\rfloor$);
  this inequality also relies on the equality $\K_{T_{i,0}}=\K_{v_{i}}$,
  which follows from our definition of the bets;
\item
  the inequality~\eqref{eq:0} holds with lower game-theoretic probability at least $1-2i\epsilon$
  by the inductive assumption;
\item
  the inequality~\eqref{eq:1} holds with lower game-theoretic probability at least $1-\epsilon$ for a large enough $L$,
  and follows from Theorem~3.1 of \cite{\CTIV}
  and the uniform continuity of $U^{\rm m}_i$ with $f$ as modulus of continuity;
\item
  the equality~\eqref{eq:1.5} holds by the definition \eqref{eq:overline-U-basis};
\item
  the inequality~\eqref{eq:2} also holds with lower game-theoretic probability at least $1-\epsilon$
  for a large enough $L$
  and follows from Theorem~3.1 of \cite{\CTIV}
  and the uniform continuity of $U^{\rm m}_i$
  with $f$ as modulus of continuity.
\end{itemize}
We can see that the overall chain \eqref{eq:chain-start}--\eqref{eq:chain-end}
holds with lower probability at least $1-2(i+1)\epsilon$.

So far we have considered the case $v_{i+1}<1-\epsilon$.
Now suppose
\begin{equation*}
  1-\epsilon\in(v_{i}(\omega),v_{i+1}(\omega)].
\end{equation*}
As soon as time $1-\epsilon$ is reached, the strategy stops playing:
we will show that with a lower game-theoretic probability arbitrarily close to $1$
the goal has been achieved.
Indeed, as we saw above,
\begin{equation*}
  \K_{v_{i}(\omega)}
  \gtrsim
  \overline{U}^{\rm e}_{i}(x_1,v_1,\ldots,x_i,v_i)
\end{equation*}
with high lower game-theoretic probability.
Let us check that
\begin{equation*}
  \K_{1-\epsilon}
  \gtrsim
  F_N(\omega)
\end{equation*}
with high lower game-theoretic probability.
This is true since $\K_{1-\epsilon}$ is, with high lower probability,
greater than or close to the average of
\begin{align*}
  \overline{U}^{\rm m}_{i}(x_1,v_1,\ldots,x_i,v_i,\xi)
  &\ge
  \overline{U}^{\rm e}_{i+1}(x_1,v_1,\ldots,x_i,v_i,\xi,1)\\
  &=
  \overline{U}^{\rm e}_{i+1}(x_1,v_1,\ldots,x_i,v_i,\omega(1),1)\\
  &=
  \overline{U}^{\rm e}_{N}(x_1,v_1,\ldots,x_i,v_i,\omega(1),1,\ldots,\omega(1),1)\\
  &\ge
  F_N(\omega) - f(\epsilon)
\end{align*}
(cf.\ \eqref{eq:F_N} and \eqref{eq:basis-U})
over the value $\xi$ at time $(i+1)S/N-\langle\omega\rangle_{1-\epsilon}$
of a Brownian motion started at $\omega(1-\epsilon)$ at time 0
and stopped when it hits level 0,
where $\langle\omega\rangle$ is the quadratic variation of $\omega$
as defined in \cite{\CTIV}, Section~8.

To ensure that his capital is always positive,
the trader stops playing as soon as his capital hits $0$.
Increasing his initial capital by a small amount
we can make sure that this will never happen
(for $L$ sufficiently large).
Increasing his initial capital by another small amount
we can make sure that he always superhedges $F_N$
and not just with high lower game-theoretic probability.
Letting $L\to\infty$, we obtain $\UEG(F_N)\le U^{\rm e}_0$.

\section{Conclusion}

There is no doubt that this version of the paper
makes various unnecessary assumptions.
To relax or eliminate those assumptions
is a natural direction of further research.

\subsection*{Acknowledgements}

The impetus for writing this paper was a series of discussions in April 2015
with Nicolas Perkowski, David Pr\"omel, Martin Huesmann,
Alexander M. G. Cox, Pietro Siorpaes, and Beatrice Acciaio
during the junior trimester ``Optimal transport''
held at the Hausdorff Mathematical Centre (Bonn, Germany).
They posed the problem of proving or disproving the coincidence
of game-theoretic and measure-theoretic probability
in the case of the full Wiener space $\Omega$,
and this paper gives a positive answer to a radically simplified version of that problem.
I am grateful to the organizers of the trimester for inviting me
to give a mini-course on game-theoretic probability.
Thanks to Gert de Cooman and Jasper de Bock for numerous discussions
and their critique of my ``narrow'' definition of game-theoretic probability
(the one given in \cite{\CTIV}) as too broad;
with apologies to them, this paper experiments with an even broader definition.

\end{document}